\documentclass[12pt]{iopart}

\usepackage{cite}
\usepackage{amsthm}
\usepackage{graphicx}
\usepackage{subfigure}
\usepackage{bm}
\usepackage{algorithm}
\usepackage{algorithmicx}
\usepackage{algpseudocode}
\usepackage{multirow}
\usepackage{color}
\usepackage{amssymb}
\newtheorem{theorem}{Theorem}

\usepackage{adjustbox}
\usepackage{diagbox}

\begin{document}

\title[Low-resolution Prior Equilibrium (LRPE) Network]{Low-resolution Prior Equilibrium Network for CT Reconstruction}

\author{Yijie Yang$^{1}$, Qifeng Gao$^{1}$ and Yuping Duan${^\ast}{^3}$}

\address{$^1$\ Center for Applied Mathematics, Tianjin University, Tianjin, 300072, China}
\address{$^3$\  School of Mathematical Sciences, Beijing Normal University No. 19, XinJieKouWai St., HaiDian District, Beijing 100875}
\address{$^\ast$\ Corresponding author}
\ead{doveduan@gmail.com}

\begin{abstract}
The unrolling method has been investigated for learning variational models in X-ray computed tomography. However, for incomplete data reconstruction, such as sparse-view and limited-angle problems, the unrolling method of gradient descent of the energy minimization problem cannot yield satisfactory results. In this paper, we present an effective CT reconstruction model, where the low-resolution image is introduced as a regularization  for incomplete data problems. In what follows, we utilize the deep equilibrium approach to unfolding of the gradient descent algorithm, thereby constructing the backbone network architecture for solving the minimization model. We theoretically discuss the convergence of the proposed low-resolution prior equilibrium model and provide the necessary conditions to guarantee its convergence. Experimental results on both sparse-view and limited-angle reconstruction problems are provided, demonstrating that our end-to-end low-resolution prior equilibrium model outperforms other state-of-the-art methods in terms of noise reduction, contrast-to-noise ratio, and preservation of edge details.
\end{abstract}

%
\vspace{2pc}

\section{Introduction}
\label{sect1}
Computed Tomography (CT) is a fundamental imaging tool that finds wide applications in various fields, such as industrial non-destructive testing, medical diagnoses, and security inspections. Specifically, let us consider the observation of a corrupted set of measurements, i.e., $\bm{b} \in \mathbb R^{N_{views}\times N_{bins}}$, obtained by applying a linear measurement operator $\bm{A}:\mathbb R^{N^2}\rightarrow \mathbb R^{N_{views}\times N_{bins}}$ to a target $\bm{u} \in \mathbb R^{N^2}$ with certain additive noises $\xi \in \mathbb R^{N_{views}\times N_{bins}}$. Mathematically, we can formulate the inverse problem as follows
\begin{equation}
\label{model1}
\bm{A}\bm{u}  + \xi = \bm{b},
\end{equation}
where $N_{{views}}$ represents the number of angles within the angular interval, $N_{{bins}}$ represents the number of units on the detector, and $N$ represents the denotes the length and width of the target image.
Our objective is to reconstruct the clean image $\bm{u}$ from the measurement $\bm{b}$. In general, when the number of pixels in a reconstructed image exceeds the number of projection samples in CT imaging, the inverse problem of \eref{model1} becomes ill-posed.

Due to the radiation, various approaches have been explored, which can be primarily categorized into two main categories, either to modify the scanning protocol by reducing the tube voltage and/or tube current \cite{meng2020semi} or downsample the measured data for CT reconstruction, such as interior CT \cite{larke2011estimated,patz2014overdiagnosis,sim2020optimal,zhao2021fast}, and sparse-view CT \cite{zhang2018sparse,yang2018low,kang2017deep,wolterink2017generative,zickert2022joint,xu2019image}.
These approaches aim to achieve dose reduction while maintaining satisfactory image quality. However, employing the first category of approaches for radiation dose reduction may result in the measured CT data being heavily contaminated with noise. The standard filtered back-projection (FBP) method, when applied to reconstruct CT images using noise-contaminated CT data, suffers from a significant degradation in image quality. On the other hand, the second category of approaches, aimed at reducing radiation dose, usually introduces additional ill-posedness into the inverse problem. It is worth noting that the regularization method is an effective technique for tackling ill-posed problems by incorporating prior information. The total variation (TV) regularization, as a prominent example, is considered state-of-the-art in handling low-dose (LD) and few-view CT \cite{li2019few}.
Moreover, alternative sparsity regularization techniques, e.g., the utilization of wavelet frames \cite{zhao2013tight}, have been investigated in CT reconstruction. A noise suppression-guided image filtering reconstruction algorithm has recently been proposed as a solution to address the low signal-to-noise ratio (SNR) problem \cite{he2020noise}.

With recent advancements in artificial intelligence and hardware performance, deep learning (DL) has emerged as a promising approach for denoising in low-dose computed tomography, which has demonstrated encouraging results. Since then, various deep learning reconstruction methods have been studied including the pre-processing in the projection domain, projection-to-image reconstruction, and post-processing in the image domain \cite{gholizadeh2020deep,greffier2020image,hata2020combination,arndt2021deep,mohammadinejad2021ct,nam2021deep,noda2021low}. These studies have shown that DL approaches consistently deliver enhanced or comparable noise suppression and maintain structural fidelity \cite{shan2019competitive,lenfant2020deep,rozema2021iterative,shin2020low,baguer2020computed}.
Broadly speaking, we can roughly divide the learning-based methods for CT reconstruction into three categories as the post-processing methods \cite{schwab2019deep,lunz2018adversarial}, the plug-and-play priors methods \cite{sreehari2016plug,chan2016plug,ryu2019plug,nair2021fixed,pesquet2021learning}
and the unrolling methods \cite{bai2019deep,gilton2021deep}. There methods have greatly improved the reconstruction qualities compared to conventional methods. The remaining issues mainly include: firstly, the theoretical property guarantees of learning-based reconstruction methods have not been fully addressed, such as the convergence of the learning-based models. Secondly, for severely ill-posed reconstruction problems, such as limited-angle and sparse-view reconstruction issues, it still requires effective learning-based solutions.

\subsection{Convergence studies of learning-based methods}

{\textbf{Post-processing methods.}} A deep neural network-based post-processing of a model-based reconstruction is parameterized as $\mathcal{P}_{\theta} \circ \mathcal{B}$, where $\mathcal{B}$ denotes a classical reconstruction method (e.g., FBP),  $\mathcal{P}_{\theta}$ represents a deep convolutional network with parameters $\bm \theta$ and $\circ$ represents composite operations. In general, there is no convergence guarantee for the post-processing strategies. That means, a small value of $\|\bm{A} \bm{u}- \bm{b}\|$ does not necessarily imply a small value of $\|\bm{A} \mathcal{P}_\theta(\bm{u})-\bm{b}\|$ for the output of $\mathcal{P}_\theta$. Such issue was addressed in \cite{schwab2019deep} by parametrizing the operator $\mathcal{P}_\theta$ as $\mathcal{P}_\theta=\bm{I}+\left(\bm{I}-\bm{A}^{\top} \bm{A}\right) \mathcal{Q}_\theta$, where $\mathcal{Q}_\theta$ is a Lipschitz-continuous deep convolutional neural network. Since $\bm{I}-\bm{A}^{\top} \bm{A}$ is the projection operator in the null-space of $\bm{A}$, the operator $\mathcal{P}_\theta$ always satisfies $\bm{A} \mathcal{P}_\theta\left(\bm{u}\right)=\bm{A} \bm{u}$. Null-space networks were demonstrated to offer convergent regularization schemes. Inspired by the theory of optimal transport, an adversarial framework for learning the regularization term was proposed in \cite{lunz2018adversarial}, which possessed the stability guarantee, subjecting to the condition that the regularization is 1-Lipschitz and coercive.

{\textbf{Plug-and-play (PnP) priors methods.}} The PnP prior approach leverages the strengths of both realistic modeling through physical knowledge and flexible learning from data patterns. Relying on the \emph{maximum a posteriori probability} (MAP) estimator, the solution of \eref{model1} can be formulated as the following minimization problem
\begin{equation}
{{\bm{u}}} \in \mathop{\arg\min}\limits_{\bm{u}} \mathcal{S}(\bm{u}, \bm{b})+\mathcal{R}(\bm{u}; \bm{\theta}),\label{problem}
\end{equation}
where $\mathcal{S}(\bm{u}, \bm{b})=-\log(p_{{\bm b}|{\bm u}}(\bm{u}))$ is the likelihood relating the solution $\bm u$ to the measurement $\bm{b}$ and $\mathcal{R}(\bm{u}; \bm{\theta})=-\log(p_{\bm u}(\bm u))$ is the denoiser to be learned by deep neural network models. The PnP priors framework models the likelihood by obeying the noise distributions of the underlying physical process, while it extracts flexible priors from data observation.  The learning-based regularization have shown excellent empirical performance, which inspired the theoretical studies on its convergence.
The first result demonstrating the global objective convergence of PnP-ADMM was presented in \cite{sreehari2016plug}, where the learned denoiser $\mathcal{R}(\bm{u}; \bm{\theta})$ should satisfy two requirements, i.e., it is continuously differentiable, and its gradient matrix is doubly stochastic. The fixed-point convergence of PnP-ADMM with a continuation scheme and bounded denoiser
was established in \cite{chan2016plug}.
Ryu \emph{et al.} \cite{ryu2019plug} demonstrated that the iterations of both PnP-proximal gradient-descent and PnP-ADMM exhibit contraction behavior when the denoiser $\mathcal{R}(\bm{u}; \bm{\theta})$ satisfies the Lipschitz condition. The fixed-point convergence of PnP-forward-backward splitting (FBS) and PnP-ADMM were proven in \cite{nair2021fixed}. Specifically, the proof was established for linear denoisers of the form $\mathcal R(\bm {u};\bm \theta) = \bm{Wu}$, where $\bm{W}$ is diagonalizable and its eigenvalues are within the range of [0, 1]. The convergence guarantees for the PnP methods were derived in \cite{hurault2021gradient} with gradient-step (GS), alleviating the need for such restrictive assumptions.
The GS denoiser, denoted as $\mathcal{R}(\bm{u}; \bm{\theta}) = \bm{I} - \nabla g_{\sigma}$, is formulated with $g_{\sigma}(\bm{u}) = |\bm{u} - {P}_{\sigma}(\bm{u})|$, where ${P}_{\sigma}$ represents a deep neural network devoid of structural constraints. The parametrization was shown to have enough expressive power to achieve state-of-the-art denoising performance in \cite{hurault2021gradient}.
In \cite{pesquet2021learning}, a closely related PnP approach was employed with the objective of providing an asymptotic characterization of the iterative PnP solutions. The key concept was to model maximally monotone operators (MMO) using a deep neural network, where the parameterization of MMOs was achieved by modeling the resolvent through a non-expansive neural network.

\begin{figure*}[t]
      \centering      \includegraphics[width=1.0\linewidth]{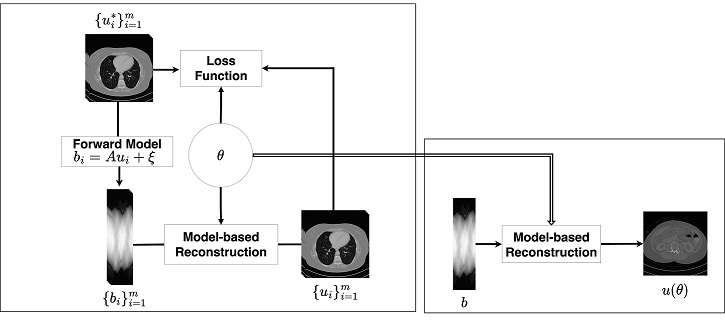}\\
    \caption{Depiction of a typical bilevel problem for image reconstruction. The left box represents the training process, which includes an upper-level loss and a lower-level cost function. During training, the objective is to minimize the upper-level loss function. Once the parameters $\bm{\theta}$ are learned, it is employed in the same image reconstruction task, as depicted in the right box.}
\label{bilevelflow}
\end{figure*}

{\textbf{Unrolling methods.}} Algorithm unrolling \cite{monga2021algorithm} unfolds the iterative algorithms with loops or recursion into non-recursive neural networks to construct interpretable and effective deep learning models.
The origin of unrolling can be traced back to the seminal work by Gregor and LeCun \cite{gregor2010learning} for solving sparse coding via unfolding the iterative soft-thresholding algorithm.
We can formally express the unrolling model for \eref{problem} as a bilevel minimization problem
\begin{equation}
\eqalign{
 \mathop{\min}\limits_{\bm{\theta}} &~~\frac{1}{m} \sum_{i=1}^m \ell\left(\bm{u}_i(\bm{\theta}), \bm{u}^*_i\right)  \quad\quad\quad\quad\quad\quad\quad\quad\quad\quad\quad\quad \hfill(\textup{UL})   \\  \mbox{s.t.}& ~~ \bm{u}_i(\bm{\theta}) \in \mathop{\arg\min}\limits_{\bm{u}}\mathcal{S}\left(\bm{u}, {b}_i\right)+\mathcal{R}(\bm{u}; \bm{\theta}) \quad\quad\quad\quad\quad~\hfill(\textup{LL})} \label{bilevel}
\end{equation}
where a loss function $\ell(\cdot,\cdot)$ is employed to enforce the similarity between the reconstructed image $\bm{u}_i$ and ground truth ${\bm{u}_i}^*$, $\mathcal{S}$ represents the data term, and $\mathcal{R}$ denotes the learnable regularization term, $m$ indicates the number of training data samples. The upper-level (UL) loss function evaluates the quality of a vector of learnable parameters, while it also depends on the solution to the lower-level (LL) cost function. Figure \ref{bilevelflow} illustrates a generic bilevel problem for image reconstruction, where the model-based reconstruction represents the process of solving the lower-level minimization problem.

The bilevel problems \cite{bard2013practical,bracken1973mathematical,colson2007overview,engl1996regularization,dempe2020bilevel} are known for their inherent difficulties to be solved numerically. To utilize gradient descent methods, it is necessary to compute the derivative of the solution operator for the lower-level problem \eref{problem} with respect to the parameters $\bm{\theta}$. If the objective of the lower-level problem is differentiable and has a unique minimum, gradients can be computed using implicit differentiation \cite{ghadimi2018approximation}. However, non-smooth objectives, along with potentially multi-valued solution operators, are frequently employed in image and signal processing tasks. The method of unrolling involves utilizing an iterative algorithm that solves the lower-level problem and replaces the optimal solution ${\bm{u}}_i(\bm\theta)$ with the $K$ stages. We can express the unrolling network as follows
\begin{equation}
\bm{u}^{(k+1)}=f_\theta^{(k)}\left(\bm{u}^{(k)}, \bm{b}; \bm{\theta}\right) \quad\mbox{for}\quad k=0, \ldots, K-1,
\end{equation}
where $k$ is the stage index and $f_\theta^{(k)}(\cdot,\bm b;\bm{\theta})$ represents a nonlinear transformation, such as convolution followed by the application of a nonlinear activation function.

In fact, different stages can share the same weights, such that $f_\theta^{(k)}(\cdot,\bm b;\bm{\theta})$ at each stage remains the same, i.e., $f_\theta^{(k)}(\cdot,\bm b;\bm{\theta}) = f_\theta(\cdot,\bm b;\bm{\theta})$ for all $k$. The equilibrium or fixed-point-based networks introduced in \cite{bai2019deep,gilton2021deep} used weight-sharing into the unrolling methods, which has demonstrated the capability to attain competitive performance. The main concept was to represent the output of a feed-forward neural network model as a fixed-point of a nonlinear transformation, which enabled the use of implicit differentiation for the back-propagation. Let us consider the $K$-stage network model with the input $\bm b$ and weights $\bm{\theta}$. The output of the $(k+1)$-th hidden layer can be denoted as $\bm{u}^{(k+1)}$ obtained by the following recursion process
\begin{equation}
\bm{u}^{(k+1)}=f_\theta\left(\bm{u}^{(k)}, \bm{b};\bm{\theta}\right).\label{equilibrium}
\end{equation}
The limit of $\bm{u}^{(K)}$ as $K \rightarrow \infty$, provided it exists, is a fixed point of the operator $f_\theta(\cdot, \bm{b};\bm{\theta})$. The series converges if the spectral norm of the Jacobian $\partial_u f_{\theta}$ is strictly less than 1 for any $\bm u$, which is true if and only if the iteration map $f_{\theta}$ is contractive. The deep equilibrium model (DEM) serves as a bridge between the conventional fixed-point methods in numerical analysis and learning-based techniques for solving the inverse problems.
We illustrate the general diagram of the unrolling algorithm based on gradient descent of the low-level minimization problem in Figure \ref{Unrolling}, while represent the structure of deep equilibrium model in Figure \ref{DEM}. As can be shown, the difference between deep unrolling model and deep equilibrium model lies in the iterative regime, where deep equilibrium model shares the weights for all stage. Indeed, we use the deep neural networks to learn the gradient operator of the data fidelity and regularization used in the gradient descent scheme.

\begin{figure*}[t]
      \centering
      \subfigure{
            \includegraphics[width=1.0\linewidth]{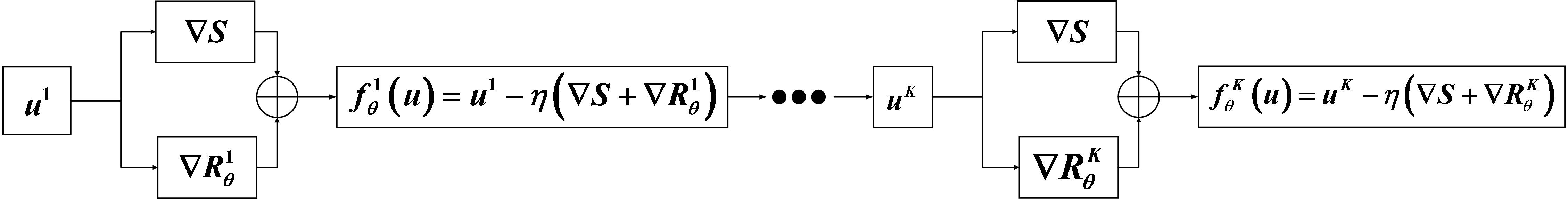}}
      \caption{Depiction of a typical unrolling method using the gradient decent.}
	\label{Unrolling}
\end{figure*}

\begin{figure*}[t]
      \centering
      \subfigure{
            \includegraphics[width=0.6\linewidth]{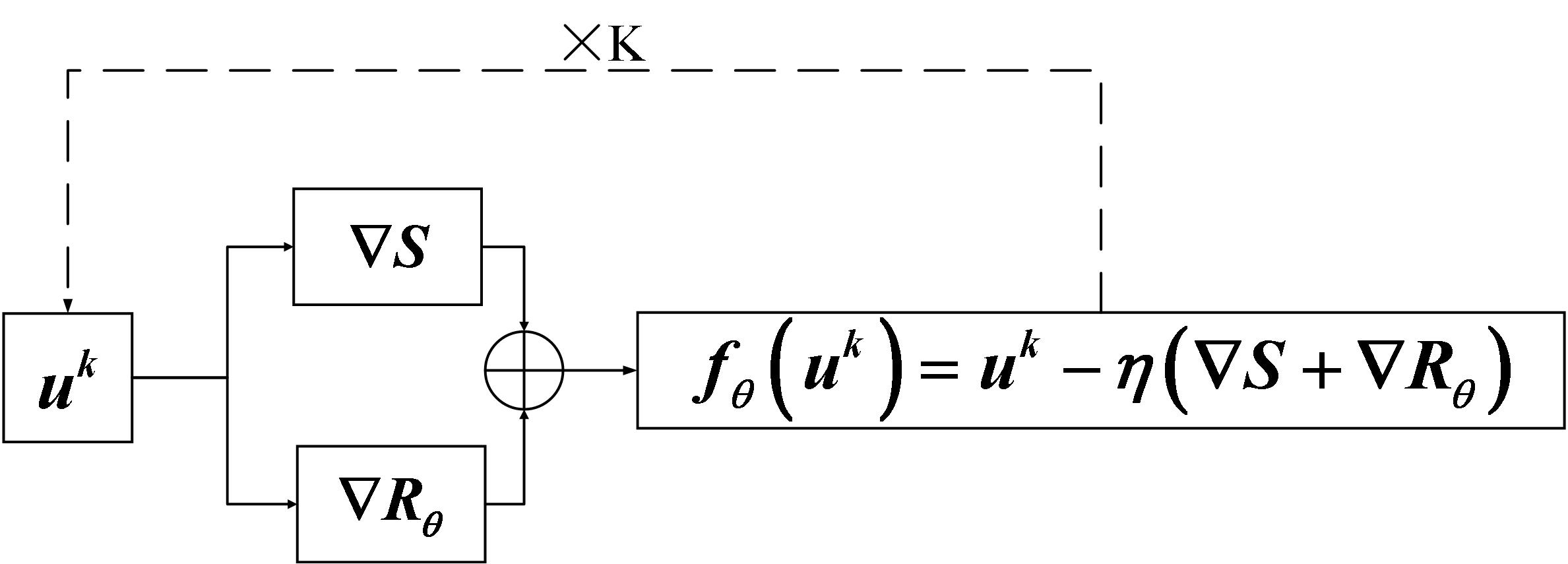}}
      \caption{Depiction of a typical deep equilibrium method using the gradient decent.}
	\label{DEM}
\end{figure*}

\subsection{Low-resolution Prior for CT reconstruction}
For severely degraded CT reconstruction problems such as sparse-view and limited-angle, it is necessary to combine effective prior information to improve the quality of the reconstructed images. Different kinds of prior knowledge have been investigated for incomplete data reconstruction problems. The non-local similarity has been shown highly beneficial for enhancing fine details in reconstructed CT images \cite{huang2011sparse,kim2016non,gilboa2009nonlocal,min2024non}. Both low-rank prior and dictionary learning have been used to leverage the sparsity of the image in a transform domain  \cite{lustig2007sparse,chen2008prior,ravishankar2010mr}. Deep learning methods by integrating geometric priors have also been studied in \cite{shen2022geometry,lin2019dudonet,wu2017iterative,zhu2018image}. Besides, personalized information as a prior has also been demonstrated to significantly enhance the reconstruction results \cite{ghosh2022towards,grandinetti2022mr}.

In CT imaging, the problem of reconstructing the object $\bm u$ from the measured sinogram $\bm b$ is ill-posed due to the finite dimension of measured rays and the infinite-dimensional nature of the unknown object $\bm u$. We can parameterize the object by utilizing a finite-series expansion as follows
\begin{equation}
\label{discrectbase}
\bm{u}(x, y)=\sum_{j=1}^{N^{2}} u_j e_j(x,y),
\end{equation}
where $\{e_j(x,y)\}$ are basis functions that should be linearly independent. The square pixels are the most commonly used basis function defined by
\[e_j(x,y) = rect\Big(\frac{x-\tilde{x}_j}{\Delta x}\Big)rect\Big(\frac{y-\tilde{y}_j}{\Delta y}\Big),\]
where
\begin{equation*}
rect(t)\triangleq 1_{\{|t|\leq1/2\}} = \left\{
\eqalign{
1 & , \quad |t|\leq1/2, \\
0 & , \quad \mbox{otherwise},
}
\right.
\end{equation*}
and ($\tilde{x}_j$,$\tilde{y}_j$) denotes the coordinates of the $j$-th pixel, ($\Delta x$, $\Delta y$) denotes the size of each pixel. The length of the $i$th ray passing through the pixel $u_j$ is denoted by $a_{ij} = \int e_j(x, y) \mathrm{d} \ell$, where all $a_{ij}$ constitute the system matrix $\bm A$ in \eref{model1}. It is well-known that using a finer grid can improve image resolution, but it also leads to an increase in the number of unknowns, which exacerbates the ill-posedness of the inverse problem and increases sensitivity to disturbances \cite{hansen1998rank}. We can achieve the goal of reducing the number of unknowns by increasing the size of pixels/voxels through modifying the basis functions  \eref{discrectbase}.
Actually, there have been studies on the impact of resolutions on CT reconstruction.  In \cite{stearns2006efficient}, a low-resolution full-view image was initially reconstructed, followed by a fine-resolution reconstruction of a specific region of interest (ROI). In particular, the image attributes of the pixels inside the ROI were set to zero, and a new projection sinogram was calculated using the forward model represented in matrix form by \eref{model1}. Afterward, the newly calculated sinogram was subtracted from the original sinogram, resulting in a sinogram that contained projections exclusively corresponding to the ROI. Dabravolski \emph{et al.} \cite{dabravolski2014multiresolution} proposed to obtain an initial reconstruction of a global coarse image, which was subsequently divided into regions comprising both fixed and non-fixed pixels. Specifically, the boundary pixels that shared at least one adjacent pixel with a different attenuation coefficient were assigned to the non-fixed pixels region. The non-fixed pixels region was comprised of boundary pixels that had at least one adjacent pixel with a different attenuation coefficient. Another method in \cite{cao2016multiresolution} involved dividing the image domain into sub-regions with varying levels of discretization, which were characterized by either a coarse or fine pixel size. The projection data obtained from both resolutions were combined to reconstruct an image with different levels of discretization, effectively incorporating information from both measurements to enhance the reconstructed image quality. Gao \emph{et al.} \cite{gao2022lrip} proposed a low-resolution image prior image reconstruction model, which was realized in an end-to-end regime showing impressive ability in dealing with limited-angle reconstruction problems. He \emph{et al.} \cite{he2021downsampled} introduced a downsampling imaging geometric modeling approach for the data acquisition process, which integrated the geometric modeling knowledge of the CT imaging system with prior knowledge obtained from a data-driven training process to achieve precise CT image reconstruction. Malik \emph{et al.} \cite{malik2021fuzzy} suggested employing image fusion to generate a composite image that combines the high spatial resolution of a partially ambiguous image obtained from incomplete data with the more dependable quantitative information of a coarser image reconstructed from the same data in an over-complete problem. The aforementioned methods demonstrate that the low-resolution image can be used as an effective prior for ill-posed CT reconstruction problems.

\subsection{Our contributions}
In this work, we introduce the low-resolution image as the prior and propose a novel image reconstruction model for sparse-view and limited-angle CT reconstruction tasks. More specifically, the low-resolution image prior can introduce effective regularity into the reconstruction model to guarantee the qualities for incomplete data problems. Subsequently, we implement the algorithm unrolling to solve our low-resolution prior image model, where the gradient descent is employed for minimizing the lower-level problem with each Jacobi block being approximated by convolutional neural networks.
To balance feature extraction capability and model size, we establish the low-resolution prior equilibrium (LRPE) network, where the weight-sharing strategy is used among all stages. More importantly, our approach has provable convergence guarantees by satisfying certain conditions to guarantee the iterative scheme converge to a fixed-point. Through extensive numerical experiments on both sparse-view and limited-angle CT reconstruction problems, our LRPE model demonstrates clear advantages over the state-of-the-art learning methods. To sum up, the major contributions of this paper are as follows
\begin{itemize}
    \item By leveraging the characteristics of the physical process of CT imaging, we use the low-resolution image as prior and establish an effective regularization model for incomplete data CT reconstruction.
    \item We employ the deep equilibrium method to solve the low-resolution prior model, where the convergence is established for both explicit data fidelity and learned data fidelity.
    \item We conduct numerical experiments to demonstrate the advantages of the low-resolution image prior and deep equilibrium strategy for incomplete data reconstruction problems.
\end{itemize}

\begin{figure*}[t]
      \centering
      \subfigure[]{
            \includegraphics[width=0.3\linewidth]{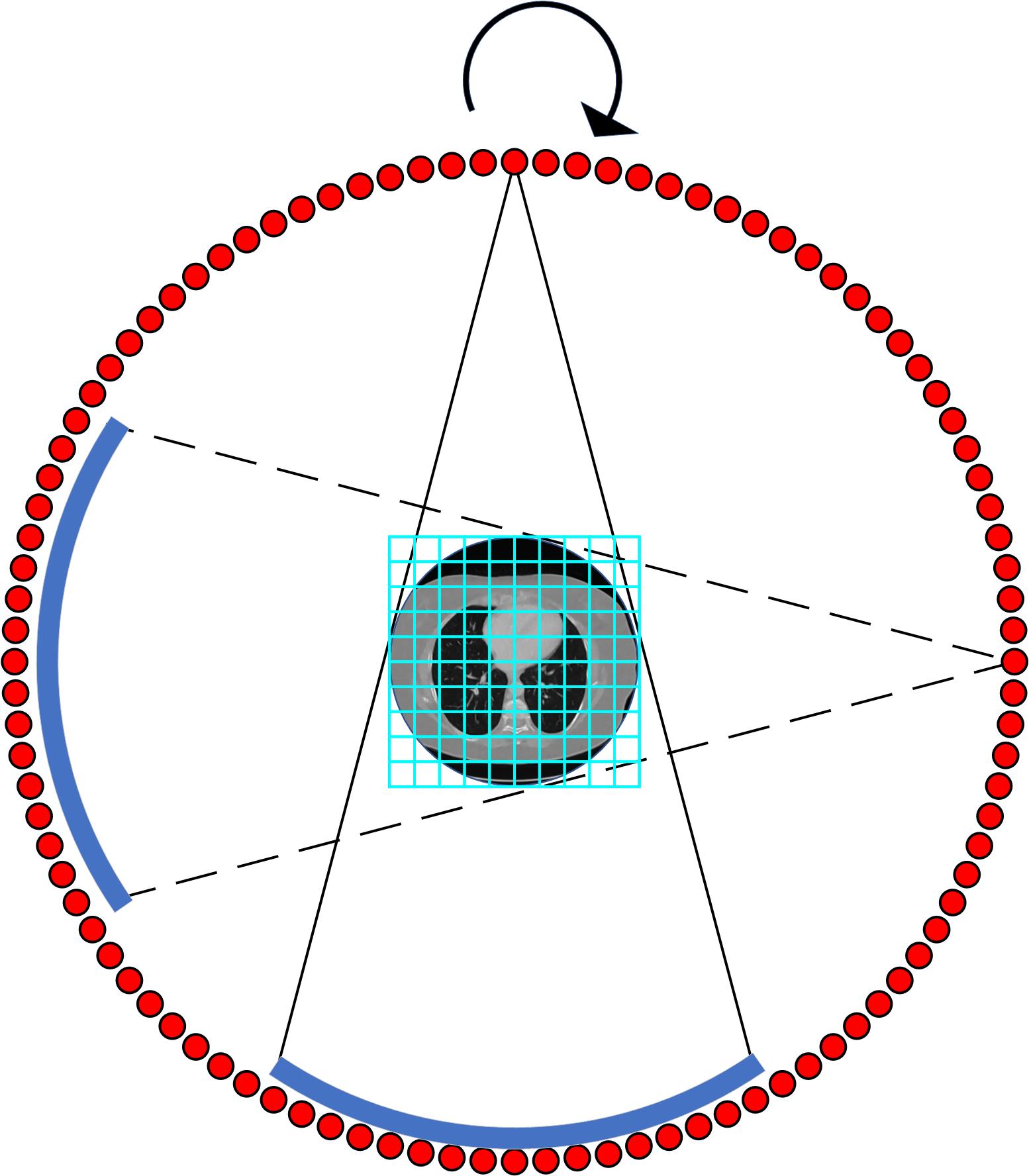}}\qquad
      \subfigure[]{
			\includegraphics[width=0.3\linewidth]{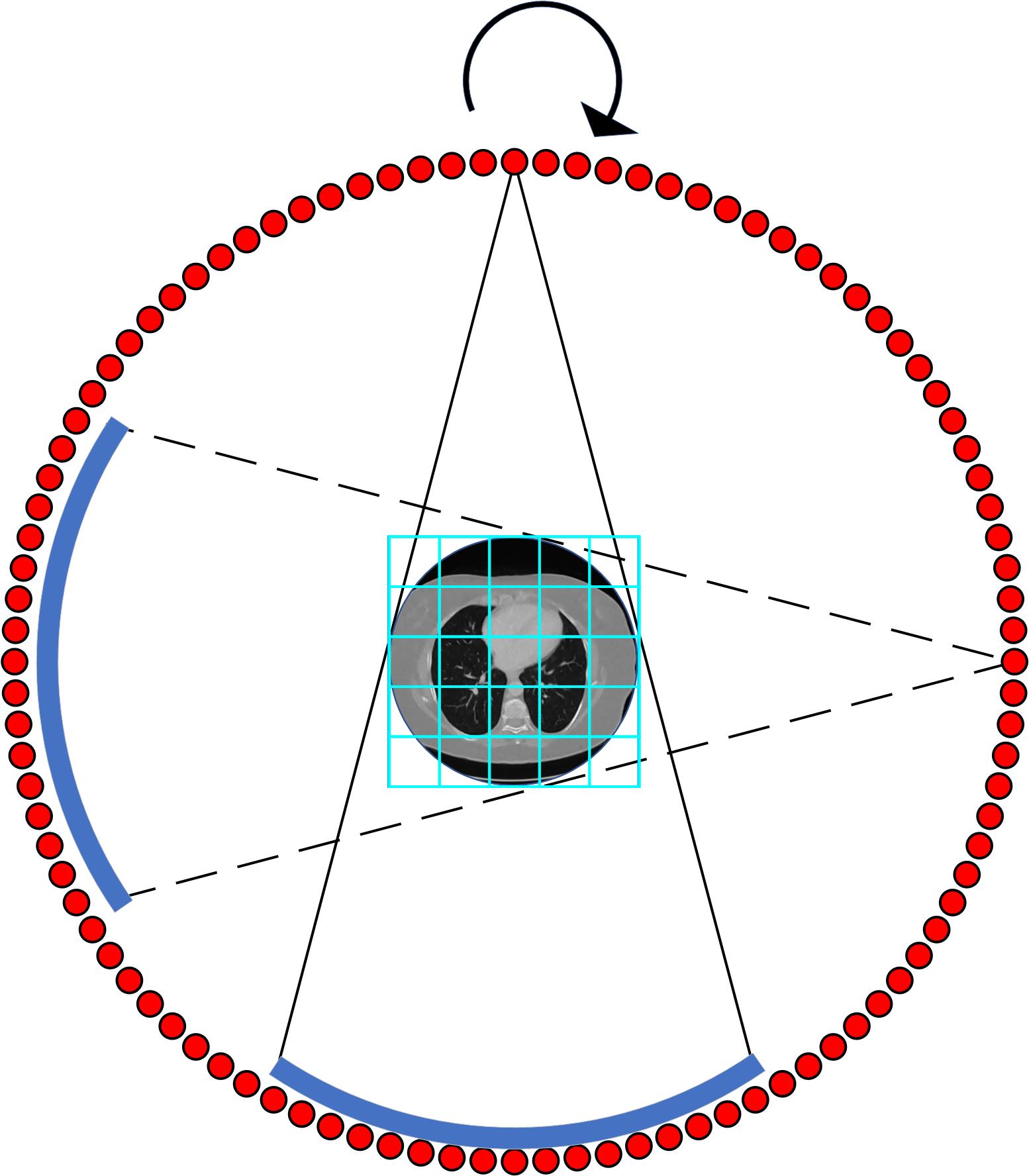}}\\
\caption{Illustration of the virtual CT scanning system, where (a) a fine CT scanning system; (b) a coarse CT scanning system.}
	\label{resolution}
\end{figure*}

\section{Our Low-resolution Prior Reconstruction Model}
\label{sect2}
Let $\bm{A_l}:\mathbb R^{n^2}\rightarrow\mathbb R^{N_{views}\times N_{bins}}$ be the system matrix for the low-resolution image $\bm{u_l}$. Thus, the low-resolution inverse problem can be formulated as follows
\begin{equation}
\label{modell}
\bm{A_l} \bm{u_l} = \bm{b},
\end{equation}
where $n$ denotes the length and width of the low-resolution image. Without considering the discretization error, the same projection data $\bm b$ is used for the low-resolution image.
The illustration of the fine and coarse CT scanning system is provided in Figure \ref{resolution}. As can be observed, the red color dots denote the permissible positions for the emitter, and the blue arc represents the corresponding receiving range of the receiver. It is obviously shown that the coarse scanning system has the same field of view (FOV) but with a larger step size in spatial discretization. That means the low-resolution image can be  obtained from the high-resolution image as follows
\begin{equation}
\label{Downsample}
\bm{u}_l = \bm{D} \bm{u},
\end{equation}
where $\bm D: \mathbb R^{N^2} \rightarrow \mathbb R^{n^2}$ is the down-sampling operator.
In this work, we concern with the following low-resolution image regularized CT reconstruction model
\begin{equation}\label{lowlevel}
{\bm{u}} \in \mathop{\arg\min}\limits_{\bm{u}} \mathcal{S}(\bm{Au}, \bm{b})+\mathcal{R}(\bm{u}, \bm{D^{\top}u}_l),
\end{equation}
where $\bm{D^\top}$ represents the transpose of the matrix $\bm{D}$. Supposing both $\mathcal{S}$ and $\mathcal{R}$ are differentiable, we can use the gradient descent to solve the minimization problem \eref{lowlevel}. That is, we start with an initial estimate $\bm{u}^{(0)}$ such as $\bm{u}^{(0)} = 0$  and a step size $\eta > 0$, such that for stage $k = 0, 1,\ldots, K-1$, we estimate $\bm u^{(k+1)}$ from
\[\bm{u}^{(k+1)}=\bm{u}^{(k)} - \eta \Big(\bm{A^{\top}} \nabla{\mathcal{S}}\left(\bm{Au}^{(k)}, \bm{b}\right)+ \nabla \mathcal{R}\left(\bm{u}^{(k)},\bm{D^{\top}u}_l\right)\Big).\]
To leverage the rich information contained in the data, we adopt the deep unrolling strategy to solve the above gradient descent problem. More specifically, we utilize two convolutional neural networks $R_{\theta}(\cdot,\cdot):\mathbb R^{N^{2}}\times \mathbb R^{N^{2}}\rightarrow \mathbb R^{N^{2}}$ and $S_\theta(\cdot,\cdot):\mathbb R^{N_{views}\times N_{bins}}\times\mathbb R^{N_{views}\times N_{bins}} \rightarrow \mathbb R^{N_{views}\times N_{bins}}$ to learn the gradient operator $\nabla\mathcal{R}\left(\bm{u}^{(k)},\bm{D^{\top}u_l}\right)$ and $\nabla \mathcal{S}\left(\bm{Au}^{(k)}, \bm{b}\right)$, respectively.
For $k = 0, \ldots, K-1$, we then have the recursive update as follows
\begin{equation}
\bm{u}^{(k+1)}=\bm{u}^{(k)} - \eta (\bm{A^{\top}}s^k+ t^k),\label{unrolling}
\end{equation}
where $s^k:= S_{\theta}(\bm{Au}^{(k)}, \bm{b})$ and $t^k:= R_\theta(\bm{u}^{(k)},\bm{D^{\top}u_l})$ are the estimations of the two convolutional neural network models, respectively.
Moreover, we use the weight-sharing scheme to improve the feature abstraction capabilities by unrolling the network to a sufficient depth.
Therefore, the final update scheme is achieved by connecting the unrolled gradient descent stages \eref{unrolling} with the deep equilibrium architecture \eref{equilibrium}, as follows
\begin{equation}
f_{\theta}(\bm{u}^{k+1}; \bm{b}):=\bm{u}^k - \eta \Big(\bm{A}^{\top}S_{\theta}(\bm{Au^k}, \bm{b})+ R_\theta(\bm{u}^k,\bm{D^{\top}u}_l)\Big).\label{equnrolling}
\end{equation}
The objective function in the upper-level problem is simply the true risk for the squared loss
\[\ell(\bm u, \bm u^*) = \frac{1}{m}\sum_{i=1}^{m}\|{{{\bm u_i}}^{*}}-{\bm u_i}^{K}\|_2^2,\]
where ${\bm u_i}^{*}$ is the $i$th ground truth, and ${\bm u}^{K}_i$ denotes the reconstructed image.
Indeed, the use of the mean squared error (MSE) loss is not necessary, as any differentiable loss function can be used.

\section{Algorithm Implementation}

The low-resolution prior equilibrium (LRPE) model can be outlined as Algorithm \ref{algorithm1}. In both training and inference stage, the low-resolution prior image $\bm u_l$ is obtained by established reconstruction methods, where we used the Learned Primal-Dual network \cite{adler2018learned} to estimate the low-resolution image in our work.

\begin{algorithm}[h]
    \caption{Low-Resolution Prior Equilibrium (LRPE) Network}
    \begin{algorithmic}[1]\\
    \textbf{Step 0.} Initialize $\bm{u}^0$, $\bm{b}$, $\bm{u_l}$, step size ${\eta}$;\\
    \textbf{Step 1.} For $k = 0,1,2,\cdots,K-1$ do
    \begin{equation*}
    \left\{\begin{array}{l}
    t^{k} \longleftarrow S_{\theta}(\bm{u}^k, \bm{D^{\top} u_l});  \\
    s^{k} \longleftarrow R_{\theta}(\bm{A u}^k, \bm{b});  \\
    \bm{u}^{k+1} = \bm{u}^k - \eta(\bm{A}^{\top}s^k + t^k);\\
    \end{array}\right.
    \end{equation*}
    \State \textbf{Step 2.} Return $\bm{u}^{K}$.
  \end{algorithmic}
  \label{algorithm1}
\end{algorithm}

In implementation, our algorithm was realized using the Operator Discretization Library (ODL), the Adler package, the ASTRA Toolbox, and TensorFlow 1.8.0. In particular, the TensorFlow is a toolkit designed for tackling complex mathematical problems, for which the calculations are represented as graphs, mathematical operations are nodes, and multidimensional data arrays are communicated as edges of the graphs. The ASTRA Toolbox is a MATLAB and Python toolbox that provides high-performance GPU primitives for 2D and 3D tomography. And ODL is a Python library focused on fast prototyping for inverse problems. Adler is a toolkit that facilitates the efficient implementation of neural network architectures.

\subsection{Network architecture}

Figure \ref{network} depicts the network structures, which is an end-to-end CT reconstruction method. There are two Resnet blocks in each stage of the low-resolution prior reconstruction. As shown at the bottom of Figure \ref{network}, each block involves a 3-layer network.
Then the overall depth of the network is determined by the number of stages it contains, which is fixed to strike a balance between the receptive fields and the total number of parameters in the network. The inclusion of residual structures \cite{he2016deep} is crucial, which can help prevent the loss of fine details, expedite information flow, and enhance the network's capacity. The residual structures can alleviate training challenges and expedite convergence. We adopt the log-student-t-distribution function as the loss function
\begin{equation}
    \phi({x}) = \frac{1}{2} \log(1+{x}^2),\label{activation}
\end{equation}
which is differentiable and suitable for modeling the statistics of natural images.
We set the numbers of channels in each stage as $6\rightarrow32\rightarrow32\rightarrow5$ for both $S_{\theta}$ and $R_{\theta}$. The convolution operator plays a vital role in feature extraction, where we set the size of the convolution kernel to 3 in implementation. The convolution stride is set to 1 and the SAME padding strategy is used in implementation. Additionally, we employ the Xavier initialization scheme for the convolution parameters and initialize all biases to zero.
\subsection{Network optimization}
In our network, the parameters are updated using the backpropagation algorithm within the stochastic gradient descent method implemented in TensorFlow. To ensure a fair comparison, we set most experimental parameters to be the same as those used in the other comparative algorithms. We employ the adaptive moment estimation (Adam) optimizer to optimize the learning rate with the parameter $\beta$ set to 0.99 and the remaining parameters set to their default values. The learning rate follows a cosine annealing strategy, which helps to improve training stability. The initial learning rate is set to $10^{-4}$ and global gradient norm clipping is performed by limiting the gradient norm to 1. Besides, the batch size is set to 1 for all experiments.

\begin{figure*}[t]
      \centering			
            \includegraphics[width=1.0\linewidth]{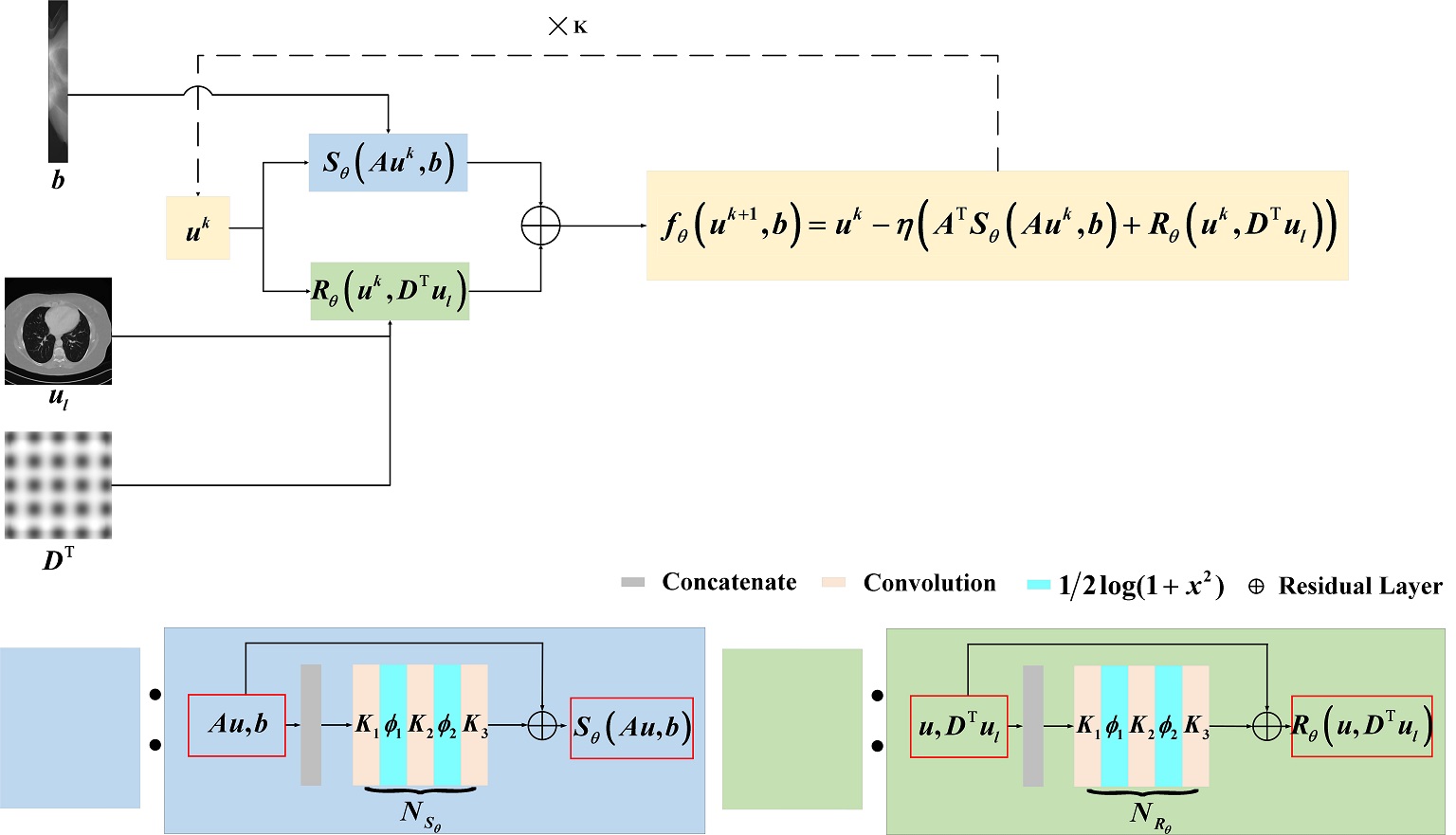}
      \caption{The network structure of our LRPE: Low-Resolution Prior Equilibrium network.}
	\label{network}
\end{figure*}

\section{Convergence Analysis}
According to the classical fixed-point theorem \cite{granas2003fixed}, the convergence of the iterates to a unique fixed-point is guaranteed if the iteration map $f_{\theta}(\cdot; \bm{b})$ is contractive, i.e. if exists a constant $0 \leq c < 1$ such that $\left\|f_{\theta}(\bm{u}; \bm{b}) - f_\theta(\bm{u}^{\prime}, \bm{b})\right\| \leq c\left\|\bm{u} - \bm{u}^{\prime}\right\|$ for all $\bm{u}$ and $\bm{u}^{\prime}$. Now we establish conditions that ensure the convergence of the iterates $\bm{u}^{(k+1)}=f_\theta\left(\bm{u}^{(k)}; \bm{b}\right)$ in Algorithm \ref{algorithm1} to a fixed-point $\bm{u}^{(\infty)}$ as the stage index $k$ approaches infinity. More specifically, we establish conditions on the regularization network $R_\theta$ and the data item network $S_\theta$, respectively. These conditions guarantee that the resulting stage map satisfies the contractivity condition and, consequently, the fixed-point iterations for these models converge.

In accordance with \cite{chan2016plug}, in the absence of noise, we make the assumptions that the regularization network $R_\theta$ and the data item network $S_\theta$ meet the following conditions: there exist positive values $\epsilon_1$ and $\epsilon_2$ such that for any $\bm{u}, \bm{u}^{\prime} \in \mathbb{R}^{N^{2}}$, the following inequality holds
\begin{equation}
\eqalign{
\left\|\left(R_\theta-\bm{I}\right)(\bm{u},\bm{D^{\top}u_l})-\left(R_\theta-\bm{I}\right)\left(\bm{u}^{\prime},\bm{D^{\top}u_l}\right)\right\| \leq \epsilon_1\left\|\bm{u}-\bm{u}^{\prime}\right\|,\\
\left\|\left(S_\theta-\bm{I}\right)(\bm{A u},\bm{b})-\left(S_\theta-\bm{I}\right)\left(\bm{A u}^{\prime},\bm{b}\right)\right\| \leq \epsilon_2\left\|\bm{u}-\bm{u}^{\prime}\right\|.}\label{lipuxis}
\end{equation}
In other words, we make the assumption that the map $R_\theta - \bm{I}$ is $\epsilon_1$-Lipschitz and $S_\theta - \bm{I}$ is $\epsilon_2$-Lipschitz. In practice, the functions $R_\theta$ and $S_\theta$ are implemented using a residual ``skip-connection" structure, such that $R_\theta = \bm{I} + N_{R_{\theta}}$ and $S_\theta = \bm{I} + N_{S_{\theta}}$. Therefore,  the inequalities in \eref{lipuxis} are equivalent to assume that the trained networks $N_{R_{\theta}}$ and $N_{S_{\theta}}$ are $\epsilon_1$-Lipschitz and $\epsilon_2$-Lipschitz, respectively. As shown in Figure \ref{network}, we have
\begin{equation*}
\eqalign{
\left\|\partial_u N_{R_{\theta}} \right\|= \left\|K_1^{\top} \partial_u\phi_1(K_1 \bm{u})K_2^{\top}\partial_u\phi_2(K_2\phi(K_1 \bm{u}))K_3\right\|,\\
\left\|\partial_u N_{S_{\theta}}\right\| = \left\|\bm{A}^{\top}K_1^{\top} \partial_u\phi_1(K_1 \bm{A u})K_2^{\top}\partial_u\phi_2(K_2\phi(K_1 \bm{A u}))K_3\right\|,}
\end{equation*}
due to $\sup _{\bm{u}}\left|\phi^{\prime}(\bm{u})\right|=0.5$, $\left\|\partial_u N_{R_{\theta}} \right\|$ and $\left\|\partial_u N_{S_{\theta}} \right\|$ can be uniformly bounded independently of $\bm{u}$. We have the following convergence result

\begin{theorem}\label{theorem1} (Convergence of LRPE with empirical data fidelity). Assume that the observed data is without noise, there is $\mathcal{S}=\frac{1}{2}||\bm{Au}-\bm{b}||_2^{2}$ in \eref{lowlevel}. Let $L=\lambda_{\max}\left(\bm{A}^{\top} \bm{A}\right)$ and $\mu=\lambda_{\min }\left(\bm{A}^{\top} \bm{A}\right)$, where $\lambda_{\max}(\cdot)$ and $\lambda_{\min }(\cdot)$ denote the maximum and minimum eigenvalue, respectively.
If the step-size parameter $\eta>0$ satisfies $\eta<1 /(L+1)$,  we can obtain the following inequality
\[\left\|f_\theta(\bm{u} ; \bm{b})-f_\theta\left(\bm{u}^{\prime} ; \bm{b}\right)\right\| \leq \underbrace{(1-\eta(1+\mu)+\eta \epsilon_1)}_{=: \gamma}\left\|\bm{u}-\bm{u}^{\prime}\right\|\]
for all $\bm{u}, \bm{u}^{\prime} \in \mathbb{R}^{N^{2}}$. The coefficient $\gamma$ is less than 1 if $\epsilon_1<1+\mu$, in which case the iterates of the LRPE model converges.
\end{theorem}
\begin{proof}
Let $f_\theta(\bm{u}; \bm{b})$ be the iteration map for the LRPE with empirical data fidelity. The Jacobian of $f_\theta(\bm{u} ; \bm{b})$ with respect to $\bm{u}$ denoted by $\partial_u f_\theta(\bm{u} ; \bm{b})$, is given by
\[\partial_u f_\theta(\bm{u} ; \bm{b})=\left(\bm{I}-\eta \bm{A}^{\top} \bm{A}\right)-\eta \partial_u R_\theta(\bm{u}),\]
where $\partial_u R_\theta(\bm{u})$ is the Jacobian of $R_\theta$ with respect to $\bm{u}$. To prove $f_\theta(\cdot ; \bm{b})$ is contractive, it suffices to show $\left\|\partial_u f_\theta(\bm{u} ; \bm{b})\right\|<1$ for all $\bm{u} \in \mathbb{R}^{N^{2}}$, where $\|\cdot\|$ denotes the spectral norm. Towards this end, we have
\begin{equation}
\scalebox{0.9}{$
\eqalign{\left\|\partial_u f_\theta(\bm{u} ; \bm{b})\right\| & =\left\|\bm{I}-\eta \bm{A}^{\top}\bm{A} -\eta \partial_u R_\theta(\bm{u})\right\| \\ & =\left\|\bm{I} -\eta \bm{A}^{\top}\bm{A} - \eta \bm{I}+\eta \bm{I} -\eta \partial_u R_\theta(\bm{u})\right\|\\ & \leq\left\|(1-\eta) \bm{I}-\eta \bm{A}^{\top} \bm{A}\right\|+\eta\left\|\partial_u R_\theta(\bm{u})-\bm{I}\right\| \\ & \leq \max _i\left|(1-\eta)-\eta \lambda_i\right|+\eta \epsilon_1,}\label{convergencegrad}
$}
\end{equation}
where $\lambda_i$ denotes the $i$ th eigenvalue of $\bm{A}^{\top} \bm{A}$. Here, we use the assumption that the map $\left(R_\theta-\bm{I}\right)(\bm{u}):=R_\theta(\bm{u})-\bm{u}$ is $\epsilon_1$-Lipschitz to guarantee the spectral norm of the Jacobian $\partial_u R_\theta(\bm{u})-\bm{I}$ to be bounded by $\epsilon_1$. By our assumption $\eta<\frac{1}{1+L}$ with $L:=\max _i \lambda_i$, we have $\eta<\frac{1}{1+\lambda_i}$ for all $i$, which implies $(1-\eta)-\eta \lambda_i>0$ for all $i$. Therefore, the maximum in \eref{convergencegrad} is obtained at $\mu:=\min _i \lambda_i$, which gives
$$
\left\|\partial_u f_\theta(\bm{u} ; \bm{b})\right\| \leq 1-\eta(1+\mu)+\eta \epsilon_1 .
$$
It shows $f_\theta$ is $\gamma$-Lipschitz with $\gamma=1-\eta(1+\mu)+\eta \epsilon_1$, proving the claim.
\end{proof}

\begin{theorem}\label{theorem2}
(Convergence of LRPE with learned data fidelity). Assume that $R_\theta-\bm{I}$ is $\epsilon_1$-Lipschitz and $S_\theta-\bm{I}$ is $\epsilon_2$-Lipschitz, and let $L=\lambda_{\max }\left(\bm{A}^{\top} \bm{A}\right)$ and $\mu=\lambda_{\min }\left(\bm{A}^{\top} \bm{A}\right)$, where $\lambda_{\max }(\cdot)$ and $\lambda_{\min }(\cdot)$ denote the maximum and minimum eigenvalue, respectively. If the step-size parameter $\eta>0$ obeys $\eta<1 /(L+1)$, the LRPE iteration map $f_\theta(\cdot ; \bm{b})$ defined in \eref{equnrolling} satisfies
\[\left\|f_\theta(\bm{u} ; \bm{b})-f_\theta\left(\bm{u}^{\prime} ; \bm{b}\right)\right\| \leq \underbrace{(1-\eta(1+\mu)+ \eta (L\epsilon_2 + \epsilon_1))}_{=: \gamma}\left\|\bm{u}-\bm{u}^{\prime}\right\|\]
for all $\bm{u}, \bm{u}^{\prime} \in \mathbb{R}^{N^{2}}$. The coefficient $\gamma$ is less than 1 if $\epsilon_1 + L\epsilon_2 <1+\mu$,  in which case the iterates of the LRPE-net converges.
\end{theorem}
\begin{proof}
Let $f_\theta(\bm{u}; \bm{b})$ be the iteration map for LRPE with learned data fidelity. The Jacobian of $f_\theta(\bm{u} ; \bm{b})$ with respect to $\bm{u} \in \mathbb{R}^{N^{2}}$, denoted by $\partial_u f_\theta(\bm{u} ; \bm{b})$, is given by
\[\partial_u f_\theta(\bm{u} ; \bm{b})=\left(\bm{I}-\eta \bm{A}^{\top} \partial_u S_\theta(\bm{Au}) \bm{A}\right)-\eta \partial_u R_\theta(\bm{u}),\]
where $\partial_u R_\theta(\bm{u})$ is the Jacobian of $R_\theta$ with respect to $\bm{u} \in \mathbb{R}^{N^{2}}$ and $\partial_u S_\theta(\bm{A u}) $ is the Jacobian of $S_\theta$. Similarly, we show $\left\|\partial_u f_\theta(\bm{u} ; \bm{b})\right\|<1$ for all $\bm{u} \in \mathbb{R}^{N^{2}}$. Since $\left(R_\theta-\bm{I}\right)(\bm{u})$ and $\left(S_\theta-\bm{I}\right)(\bm{u})$ are $\epsilon_1$-Lipschitz and $\epsilon_2$-Lipschitz, respectively, we have the spectral norm of its Jacobian $\partial_u R_\theta(\bm{u})-\bm{I}$ is bounded by $\epsilon_1$ and $\partial_u S_\theta(\bm{u})-\bm{I}$ is bounded by $\epsilon_2$. Then there is
\begin{equation}
\scalebox{0.9}{$
\eqalign{\left\|\partial_u f_\theta(\bm{u} ; \bm{b})\right\| & =\left\|\bm{I}-\eta \bm{A}^{\top} \partial_u S_\theta(\bm{Au}) \bm{A} -\eta \partial_u R_\theta(\bm u)\right\| \\ & =\left\|\bm{I} - \eta \bm{I}+\eta \bm{I} -\eta \bm{A}^{\top} \bm{A} + \eta \bm{A}^{\top} \bm{A} -\eta \bm{A}^{\top} \partial_u S_\theta(\bm{Au}) \bm{A} -\eta \partial_u R_\theta(\bm{u})\right\| \\ & =\left\| \bm{I}-\eta \bm{A}^{\top} \bm{A} - \eta \bm{I} -\eta\left(\partial_u R_\theta(u)-\bm{I}\right) - \eta \bm{A}^{\top}(\partial_u S_\theta(\bm{Au}) - \bm{I})\bm{A} \right\| \\ & \leq\left\|(1-\eta) \bm{I}-\eta \bm{A}^{\top} \bm{A}\right\|+\eta\left\|\partial_u R_\theta(\bm{u})-\bm{I}\right\| + \eta \bm{A}^{\top}\left\|\partial_u S_\theta(\bm{Au})-\bm{I}\right\|\bm{A} \\ & \leq \max _i \|(1-\eta)-\eta \lambda_i\|+\eta \epsilon_1 + \eta L \epsilon_2,}\label{convergence}
$}
\end{equation}
where $\lambda_i$ denotes the $i$ the eigenvalue of $\bm{A}^{\top} \bm{A}$. Finally, by our assumption $\eta<\frac{1}{1+L}$ where $L:=\max _i \lambda_1$, we have $\eta<\frac{1}{1+\lambda_i}$ for all $i$, which implies $(1-\eta)-\eta \lambda_i>0$ for all $i$. Therefore, the maximum in \eref{convergence} is obtained at $\mu:=\min _i \lambda_i$, which gives
\[\left\|\partial_u f_\theta(\bm{u} ; \bm{b})\right\| \leq 1-\eta(1+\mu)+\eta \epsilon_1 + \eta L \epsilon_2.\]
It shows that $f_\theta$ is $\gamma$-Lipschitz with $\gamma=1-\eta(1+\mu)+ \eta (L\epsilon_2 + \epsilon_1)$ , proving the claim.
\end{proof}

\section{Numerical Results}
In this section, we evaluate our LRPE model on both sparse-view and limited-angle reconstruction problems and compare it with several state-of-the-art methods using a dataset of human phantoms. We utilize the peak signal-to-noise ratio (PSNR) and the structural similarity index (SSIM) as evaluation metrics to measure the quality of the reconstructed images produced by different methods.

\subsection{Comparison algorithms}
We employ several established CT reconstruction methods for comparison studies, including both variational and learning-based approaches, as described below:
\begin{enumerate}
\item[$\bullet$]TV model: the TV regularized reconstruction model proposed in \cite{cai2014edge}. We tuned the balance parameter $\lambda\in [1, 3]$, the step size for the primal value $\tau$ within the range of [0.5, 0.9], and the step size for the dual value $\sigma$ within the range of [0.2, 0.5]. These parameter settings were adjusted accordingly for different experiments.
\item[$\bullet$]PD: the Learned Primal-Dual network in \cite{adler2018learned}. The network is a deep unrolled neural network with 10 stages. The number of initialization channels for both primal and dual values is set to 5. The network parameters are initialized using the Xavier initialization scheme. In all experiments, we employed the mean squared loss as the objective function, which measures the discrepancy between the reconstructed image and the ground truth.
\item[$\bullet$]SIPID: the Sinogram Interpolation and Image Denoising (SIPID) network presented in \cite{yuan2018sipid}. The SIPID network utilizes a deep learning framework and achieves accurate reconstructions by iteratively training the sinogram interpolation network and the image denoising network. The network parameters are initialized using the Xavier initialization scheme, and the mean squared loss is employed as the objective function in all experiments.
\item[$\bullet$]FSR: the Learned Full-Sampling Reconstruction From Incomplete Data in \cite{cheng2019learned}. The FSR-Net is an iterative expansion method that uses the corresponding full sampling projection system matrix as prior information. They employed two separate networks, namely IFSR and SFSR. Specifically, the IFSR and SFSR are utilized for reconstructions using IFS and SFS system matrices, respectively. The number of initialization channels for primal values and dual values is set as 6 and 7, respectively. The loss function is the mean square error of the image domain and the Radon domain with the weight $\alpha$ being 1.
\item[$\bullet$]LRIP: the Low-Resolution Image Prior based Network in \cite{gao2022lrip}. It is a low-resolution image prior image reconstruction model for the limited-angle reconstruction
problems. The number of initialization channels is set as 5 for both primal and dual values. The loss function used in training is a combination of mean squared error (MSE) and structural similarity index (SSIM) calculated in the image domain. The weight parameter $\alpha$ for balancing the two components of the loss function is set to 1.
\item[$\bullet$]GRAD: gradient descent is used instead of the primal-dual algorithm to solve the low-level problem described in equation \eref{lowlevel}. It does not incorporate low-resolution image prior information. The network is a deep unrolled neural network with 10 stages. The number of initialization channels is set to 5. The Xavier initialization and the mean squared loss of the reconstructed image and the ground truth are used in all experiments.
\end{enumerate}

\subsection{Datasets and settings}
In the experiments, we utilize the clinical dataset known as ``The 2016 NIH-AAPM-Mayo Clinic Low Dose CT Grand Challenge" \cite{mccollough2016tu}. This dataset comprises 10 full-dose scans of the ACR CT accreditation phantom. To establish the training dataset, we select 9 of these scans, reserving the remaining 1 scan for evaluation purposes. Consequently, the training dataset consists of 2164 images, each with dimensions of 512$\times$512, while the evaluation dataset comprises 214 images. We set the scanning angular interval to be 1 degree. To assess the performance of the reconstruction methods, we introduce various types of noises into the projected data, which allows us to validate and compare the effectiveness of the different reconstruction techniques.

\begin{figure*}[htbp!]
      \centering
\includegraphics[width=0.45\linewidth]{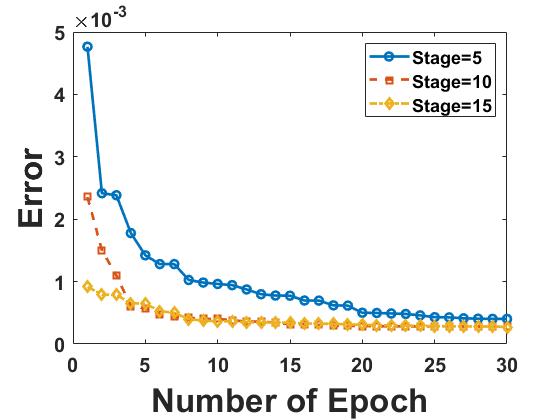}
\includegraphics[width=0.45\linewidth]{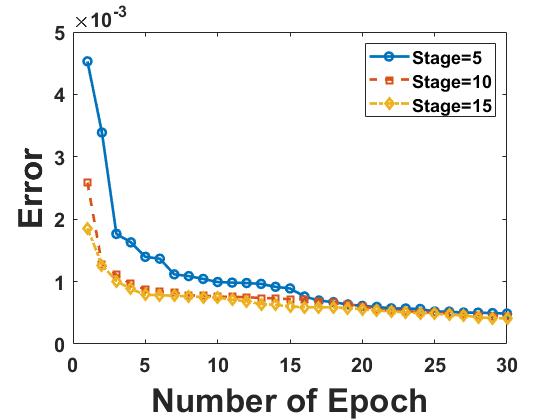}\\
\includegraphics[width=0.45\linewidth]{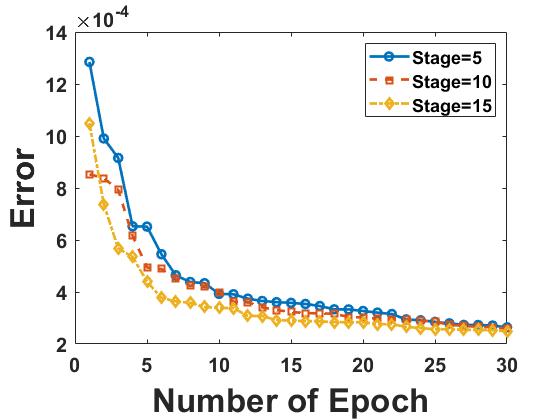}
\includegraphics[width=0.45\linewidth]{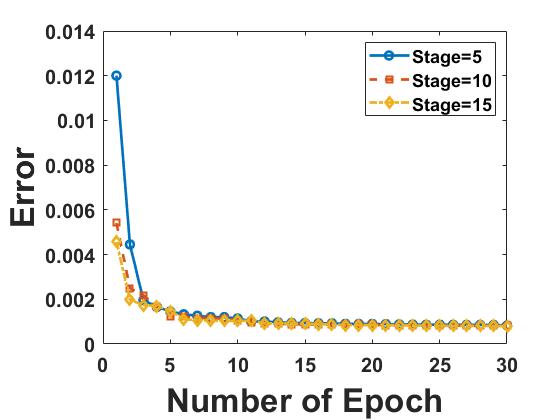}
	\caption{The data error measured by  L2 norm with respect to the numbers of epochs in our LRPE model. The curves represent the results for $150^\circ$ limited-angle (First row)and 60 sparse-view (Second row) with and without 5\% Gaussian noises.  On each row, the left  one is the result without noises, and the right one is the result with 5\% Gaussian noises.}
	\label{curve}
\end{figure*}

\subsection{Test Settings and Parameter Choice}
\label{sect4}
In this subsection, we evaluate the choices of the parameters to the performance of our algorithm.
There are several important parameters in  Algorithm \ref{algorithm1}, including the number of stages $K$, the step size $\eta$ and the dimension of the variables in the network model.

We first study the influence of the stage number $K$ on the convergence of the algorithm. We use the squared error $\|\bm u^K-\bm u^*\|^2$ to measure the data error during the training process. Both the results of $150^\circ$ limited-angle and 60 sparse-view reconstruction with and without 5\% Gaussian noises are presented in Figure \ref{curve}. As we can see, our LRPE model demonstrates convergence as the number of epochs increases, regardless of whether there are 5, 10, or 15 stages (iterations). As Theorem 4.8 in \cite{bottou2018optimization} indicates, convergence to the optimal value can be achieved when the data itself is noise-free. If the observed data contains noises, the expected objective value of the fixed-point iteration converges to a neighborhood of the optimal value. We can observe from the numerical error of our LRPE method converges to lower error for cases without noises. When noises are introduced into observed data, the solution converges to much higher error, which is consistent with the theoretical analysis results in \cite{bottou2018optimization}. With the same number of iterations, more stages can optimize to a lower data error faster. Since there
is a trade-off between the error reduction and training efficiency, we fix
the total number of stages to 10 to achieve a good balance.

On the other hand, we also conducted experiments with different numbers of stages during the inference phase, where the LRPE models trained with 10 stages are used in the experiment. As can be observed in Figure \ref{parameters}(a), as the number of stages increases, the PSNR generally improves for all the limited-angle reconstruction problems.
It is worth noting that by using a parameter-sharing strategy in the deep equilibrium model, the number of parameters is significantly reduced compared to the learned PD model. Specifically, the number of parameters in our model is only 1/10th of the number of parameters in the learned PD model. In what follows, we analyze the effect of the hyper-parameter $\eta$ on the performance of limited-angle reconstruction. We perform a parameter sweep and obtain the results shown in Figure \ref{parameters}(b), which models are used 10 stages. From the results, we observe that the highest PSNR value is achieved when we set $\eta$ as $\eta = 0.1$. Therefore, in all other experiments, we fix the hyper-parameter $\eta$ to a value of 0.1.

\begin{figure*}[htbp!]
      \centering
      \subfigure[Results of the different stages]{
            \includegraphics[width=0.45\linewidth]{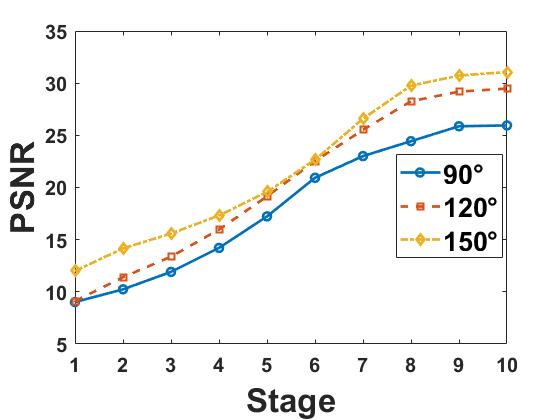}}
      \subfigure[Results of the different step size]{
            \includegraphics[width=0.45\linewidth]{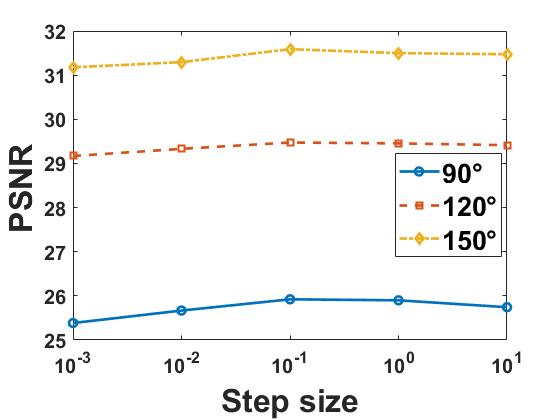}}
      \caption{Evaluation results of our LRPE model with respect to different settings of parameters during inference process.}
	\label{parameters}
\end{figure*}

Expanding the variable space is a common technique used in network optimization to improve stability during the training process. In our case, we expand the variable space by considering multiple variables $\bm{u} = [\bm{u}^{(1)}, \bm{u}^{(2)}, \ldots, \bm{u}^{(N_p)}]$. In Table \ref{dimention}, we investigate the influence of different choices of $N_p$ on the reconstruction accuracy using limited-angle data. From the perspective of results obtained from various views in Table \ref{dimention}, we fix the value of $N_p$ to 5 for all experiments, which provides a good balance between reconstruction accuracy and computational efficiency.
\begin{table*}[t]
\caption{The performance of our LRPE was evaluated in terms of PSNR across different values of $N_p$.}
\centering
\label{dimention}
\begin{tabular}{c|c|c|c|c|c|c}
\hline
\hline
\small \diagbox{Limited-angle}{$N_p$}&1&3&4&5&6&7\\
\hline
\multirow{1}{*}{$90^{\circ}$} &\small 25.6472&\small 25.8506 &\small  25.9506 &\small 25.9517 &\small  25.9129 &\small  \bf{25.9565}   \\
\hline
\multirow{1}{*}{$120^{\circ}$} &\small 29.1946&\small  29.2797 &\small  29.4326 &\small  \bf{29.4835} &\small  29.4718 &\small  29.4815  \\
\hline
\multirow{1}{*}{$150^{\circ}$} &\small 31.2969&\small 31.4487 &\small  31.5884 &\small  \bf{31.636} &\small  31.4477 &\small  31.3321   \\
\hline
\hline
\end{tabular}
\end{table*}

\subsection{Empirical data fidelity or learned data fidelity?}
We evaluate the differences between the empirical data fidelity and learned fidelity on 150$^\circ$ limited-angle reconstruction problems, where the empirical one is chosen as the L2 norm following the assumption in Theorem \ref{theorem1} and the learned one following the assumption in Theorem \ref{theorem2}. Table \ref{mixed} presents the PSNR and SSIM values for both cases, for which the raw data are corrupted by 5$\%$ Gaussian noises and Poisson noises with up to 100, 1000, and 10000 incident photons per pixel before attenuation, respectively. As can be seen, when the mixed noises are contained in the projection data, the model with the learned fidelity outperforms the model with the empirical fidelity. It is attributed to the fact that the L2 norm measures the Euclidean distance, which is not ideal for complex noise distributions in mixed noises. When salt-and-pepper noises are introduced into the measured data, we can observe similar results as shown in Table \ref{salt}. To sum up, the learned data fidelity performs better than the empirical data fidelity term, especially when the noises become significant.
\begin{table*}[!htbp]
\caption{Evaluation results on the $150^{\circ}$ limited-angle reconstruction problem corrupted by the mixed Gaussian and Poisson noises.}
\centering
\label{mixed}
\begin{tabular}{c|c|c|c|c|c|c}
\hline
\hline
    \multirow{2}{*}{\diagbox{Method}{Settings}}      & \multicolumn{2}{c|}{100 photons} & \multicolumn{2}{c|}{1000 photons} & \multicolumn{2}{c}{10000 photons}\\
    \cline{2-7}
    & PSNR  & SSIM  & PSNR  & SSIM & PSNR  & SSIM\\
    \cline{1-7}
    Empirical & 29.6282 & 0.9172    &  30.4864     & 0.9082 &  31.5067     & 0.9412\\
    \cline{1-7}
    Learned & \bf{29.8253} & \bf{0.9194}    &  \bf{30.9255}     & \bf{0.9193} &  \bf{32.2568}     & \bf{0.9446}\\
\hline
\hline
\end{tabular}
\end{table*}

\begin{table*}[!htbp]
\caption{Evaluation results on the  $150^{\circ}$ limited-angle reconstruction problems corrupted by salt-and-pepper noises.}
\centering
\label{salt}
\begin{tabular}{c|c|c|c|c|c|c}
\hline
\hline
    \multirow{2}{*}{\diagbox{Method}{Settings}}      & \multicolumn{2}{c|}{10\% noise} & \multicolumn{2}{c|}{5\% noise} & \multicolumn{2}{c}{1\% noise}\\
    \cline{2-7}
    & PSNR  & SSIM  & PSNR  & SSIM & PSNR  & SSIM\\
    \cline{1-7}
    Empirical & 33.6648 & 0.9412    &  33.9457     & 0.9436 &  34.5948     & 0.9493\\
    \cline{1-7}
    Learned & \bf{33.9674} & \bf{0.9508}    &  \bf{34.4504}     & \bf{0.9586} &  \bf{35.1745}     & \bf{0.9589}\\
\hline
\hline
\end{tabular}
\end{table*}

\begin{figure*}[t]
      \centering
      \subfigure[Clean]{
			\includegraphics[width=0.23\linewidth]{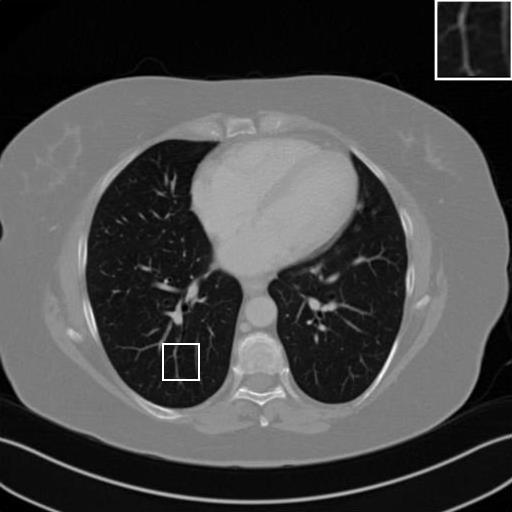}}
      \subfigure[FBP]{
            \includegraphics[width=0.23\linewidth]{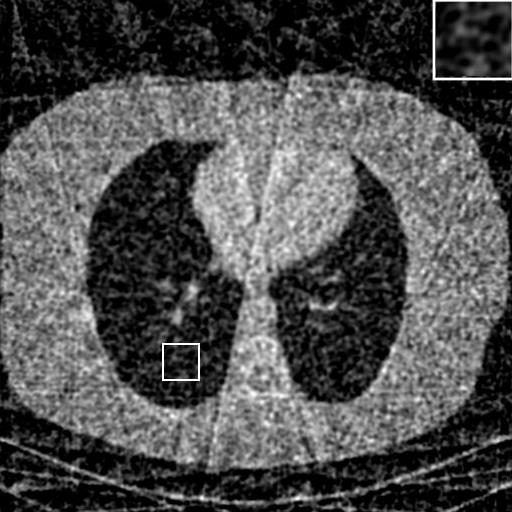}}
      \subfigure[TV]{
			\includegraphics[width=0.23\linewidth]{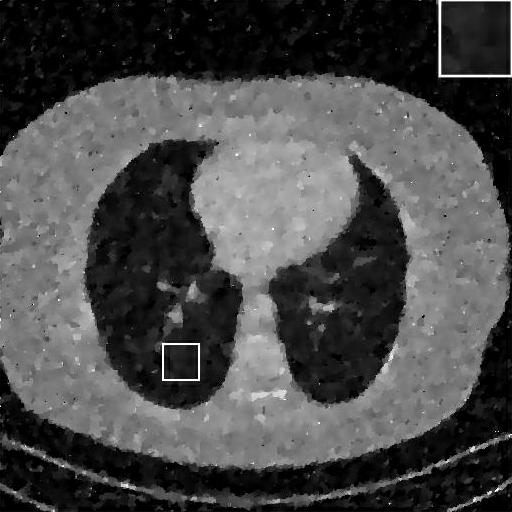}}
      \subfigure[PD-net]{
           \includegraphics[width=0.23\linewidth]{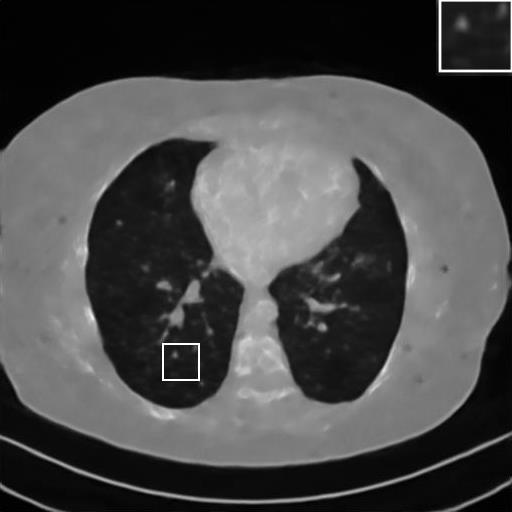}}\\
      \subfigure[IFSR]{
			\includegraphics[width=0.23\linewidth]{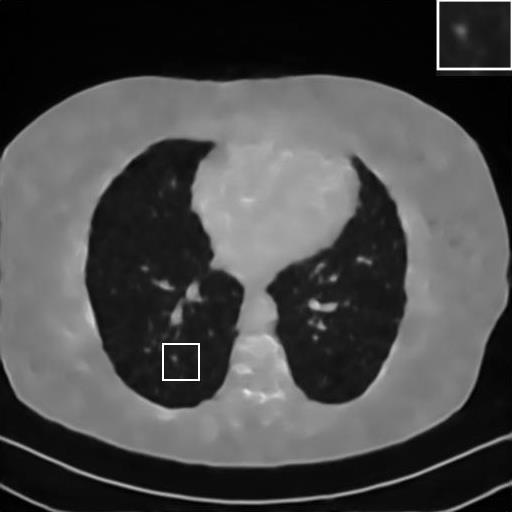}}
      \subfigure[SFSR]{
            \includegraphics[width=0.23\linewidth]{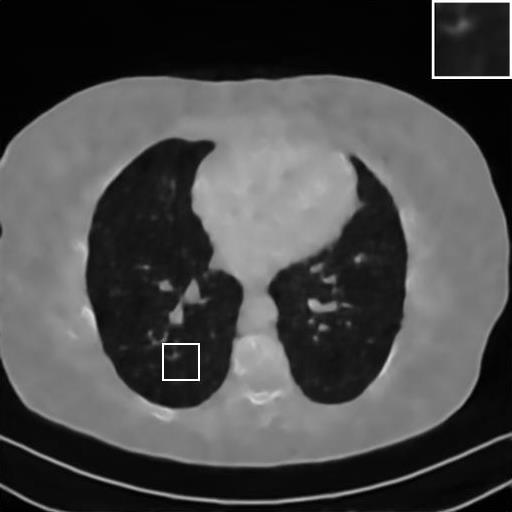}}
      \subfigure[GRAD]{
			\includegraphics[width=0.23\linewidth]{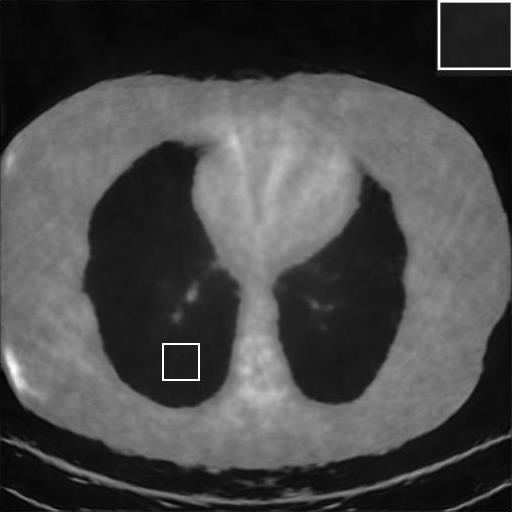}}
      \subfigure[LRPE]{
           \includegraphics[width=0.23\linewidth]{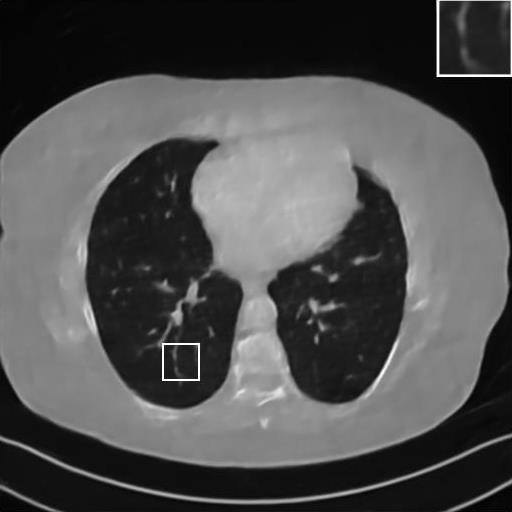}}\\
\caption{The sparse-view reconstruction experiments are performed on the AAPM phantom dataset using 30 views and 5\% Gaussian noise. The display window was set to [0, 1].}
	\label{30}
\end{figure*}

\subsection{Experiments on the sparse-view reconstruction}
We further evaluate the image qualities of our LRPE
model under the low-dose sparse-view projection data. We utilize 60 views, 45 views, and 30 views projection data corrupted by 5\% white Gaussian noises for evaluation, where the reconstruction results of different reconstruction methods are provided in Table \ref{sparse}. Both PSNR and SSIM demonstrate that our LRPE model outperforms other reconstruction methods. Figure \ref{30} presents the reconstructed images obtained by the comparative methods on the 30-view reconstruction problem. By comparing images (g) and (h), it is evident that the low-resolution image prior can significantly improve the reconstruction quality by preserving fine details and sharp edges. Moreover, the zoomed region exhibits that the deep equilibrium architecture can well maintain structural information.

We further increase the noise level in the raw data to 10\% white Gaussian noises and provide the quantitative results in Table \ref{sparse10}. It can be observed that the FBP method performs poorly in the presence of high-level noises, with a decrease in PSNR by 4 dB compared to previous experiments. In contrast, the learning-based methods are demonstrated to be less sensitive to noises. Similarly, our LRPE model exhibits the best performance among all the deep learning-based algorithms.
\begin{table}[htbp]
  \centering
  \caption{Comparison of different methods on sparse-view data corrupted by $5\%$ Gaussian noises in terms of PSNR and SSIM.}
    \begin{tabular}{c|c|c|c|c|c|c|c|c}
    \hline
    \hline
    \small $N_{view}$ & \small Metrics & \small{FBP} &\small{TV} &  \small{PD} & \small{IFSR} & \small{SFSR} & \small{GRAD} & \small{LRPE}\\
    \hline
    \multirow{2}[0]{*}{\small 60} & \small{PSNR}  & 21.4443 & 27.0301 & 30.0261 & 30.4563 & 30.4681 & 26.2635 & \bf{31.0746} \\
    \cline{2-9}
          & \small SSIM  & 0.4828 & 0.8548 & 0.9237 & 0.9321 & 0.9331 & 0.8647 & \bf{0.9337} \\
          \hline
    \multirow{2}[0]{*}{\small 45} & \small PSNR  & 19.7123 & 25.7201  & 29.1028 & 29.6814 & 29.7915 & 24.3606 & \bf{30.8035} \\
    \cline{2-9}
          & \small SSIM  & 0.4061 & 0.8183  & 0.9135 & 0.9253 & 0.9263 & 0.8031 & \bf{0.9342} \\
          \hline
    \multirow{2}[0]{*}{\small 30} & \small PSNR  & 17.8999 & 23.705  & 28.9207 & 29.0632 & 29.5072 & 23.713 & \bf{29.7237} \\
    \cline{2-9}
          & \small SSIM  & 0.3257 & 0.7615  & 0.9167 & 0.9198 & \bf{0.9244} & 0.8355 & 0.921 \\
    \hline
    \hline
    \end{tabular}
  \label{sparse}
\end{table}

\begin{table}[htbp]
  \centering
  \caption{Comparison of different methods on sparse-view data corrupted by $10\%$ Gaussian noises in terms of PSNR and SSIM. }
    \begin{tabular}{c|c|c|c|c|c|c|c|c}
    \hline
    \hline
    \small $N_{view}$ & \small Metrics & \small{FBP} &\small{TV} &  \small{PD} & \small{IFSR} & \small{SFSR} & \small{GRAD} & \small{LRPE} \\
    \hline
    \multirow{2}[0]{*}{\small 60} & \small PSNR  & 17.6258 & 24.6045 & 28.311 & 28.3782 & 28.6113 & 24.4988 & \bf{28.7754} \\
    \cline{2-9}
          & \small SSIM  & 0.277 & 0.7826 & 0.9079 & 0.9147 & 0.9161 & 0.8514 & \bf{0.9174} \\
          \hline
    \multirow{2}[0]{*}{\small 45} & \small PSNR  & 16.1392 & 23.7815 & 27.5534 & 28.0502 & 28.3181 & 23.6734 & \bf{28.6671} \\
    \cline{2-9}
          & \small SSIM  & 0.2236 & 0.752 & 0.9049 & 0.9116 & \bf{0.9141} & 0.8318 & 0.9137 \\
          \hline
    \multirow{2}[0]{*}{\small 30} & \small PSNR  & 14.3605 & 22.9388 & 26.9146 & 27.0682 & 27.653 & 23.0145 & \bf{27.8697} \\
    \cline{2-9}
          & \small SSIM  & 0.166 & 0.7416 & 0.9011 & 0.9026 & \bf{0.9071} & 0.8156 & 0.9043 \\
    \hline
    \hline
    \end{tabular}
  \label{sparse10}
\end{table}

\begin{table}[t]
\footnotesize
  \centering
  \caption{Comparison of different methods on limited-angle data corrupted by $5\%$ Gaussian noises in terms of PSNR and SSIM.}
    \begin{tabular}{c|c|c|c|c|c|c|c|c|c}
    \hline
    \hline
    \small $N_{view}$ & {\small Metrics} & \small{FBP} &\small{TV} &  \small{SIPID} & \small{PD} & \small{SFSR} & \small{LRIP} & \small{GRAD} & \small{LRPE} \\
    \hline
    \multirow{2}[0]{*}{\small $150^{\circ}$} & \small PSNR  & 13.5911 & 25.8815 & 30.3275 & 30.3766 & 30.9411 & 31.5957 & 24.3043 & \bf{31.636} \\
    \cline{2-10}
          & \small SSIM  & 0.4854 & 0.8091 & 0.9276 & 0.9301 & 0.9324 & \bf{0.9426} & 0.8572 & 0.9422 \\
          \hline
    \multirow{2}[0]{*}{\small $120^{\circ}$} & \small PSNR  & 13.4418 & 23.5852 & 27.0428 & 27.1539 & 28.3263 & 29.2763 & 21.6601 & \bf{29.4835} \\
    \cline{2-10}
          & \small SSIM  & 0.4008 & 0.7891 & 0.9024 & 0.9037 & 0.9103 & 0.9361 & 0.801 & \bf{0.9378} \\
          \hline
    \multirow{2}[0]{*}{\small $90^{\circ}$} & \small PSNR  & 13.0314 & 19.9501  & 22.7492 & 22.6047 & 24.2494 & 25.1555 & 19.3469 & \bf{25.9517} \\
    \cline{2-10}
          & \small SSIM  & 0.3881 & 0.6918  & 0.8626 & 0.8612 & 0.8761 & \bf{0.8893} & 0.7538 & 0.8814 \\
    \hline
    \hline
    \end{tabular}
  \label{limite}
\end{table}
\begin{figure*}[h]
      \centering
      \subfigure[FBP]{
            \includegraphics[width=0.23\linewidth]{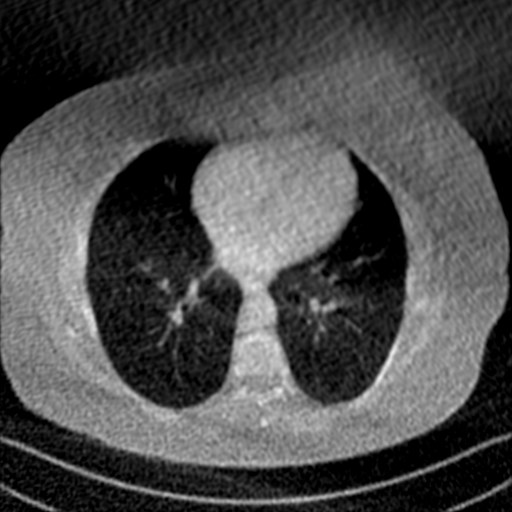}}
      \subfigure[TV]{
			\includegraphics[width=0.23\linewidth]{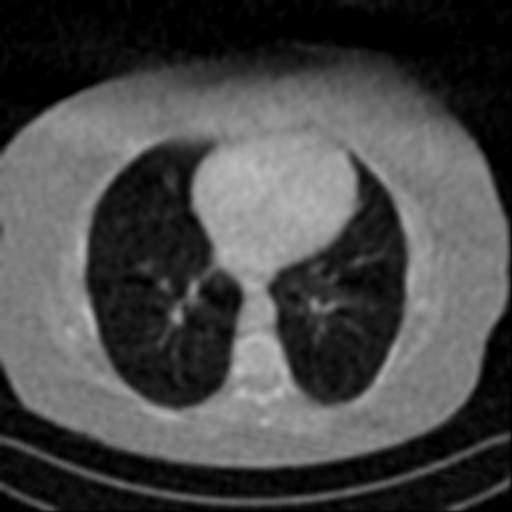}}
      \subfigure[SIPID]{
            \includegraphics[width=0.23\linewidth]{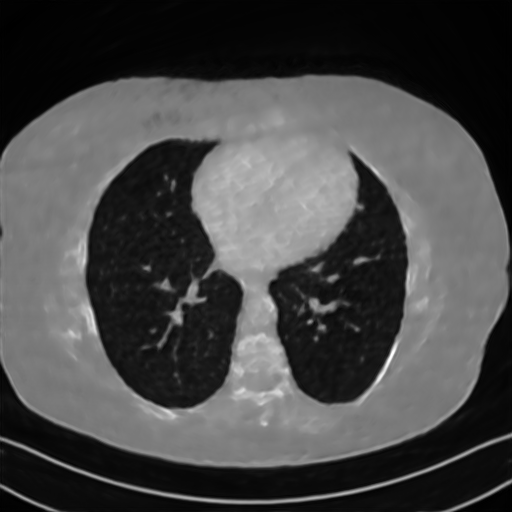}}
      \subfigure[PD]{
			\includegraphics[width=0.23\linewidth]{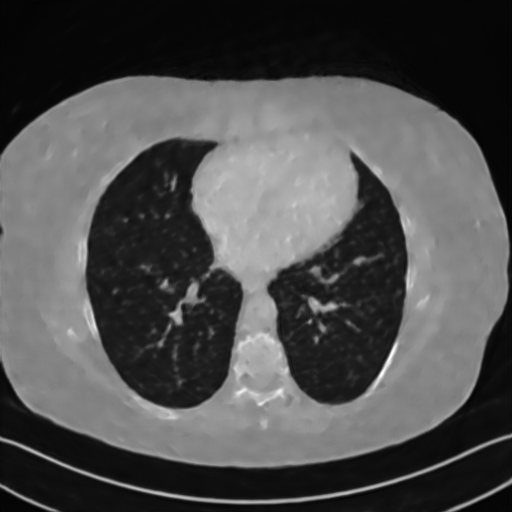}}\\
      \subfigure[SFSR]{
            \includegraphics[width=0.23\linewidth]{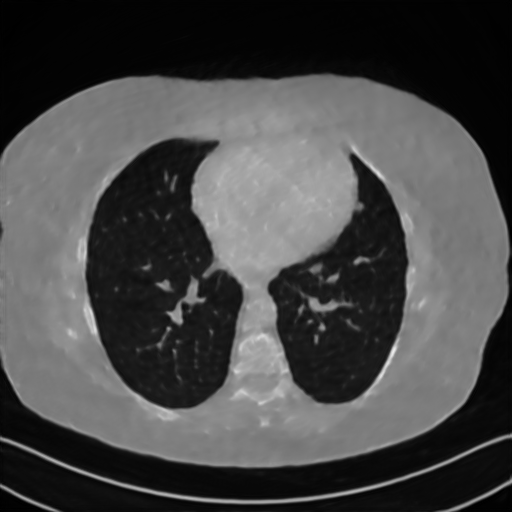}}
      \subfigure[LRIP]{
			\includegraphics[width=0.23\linewidth]{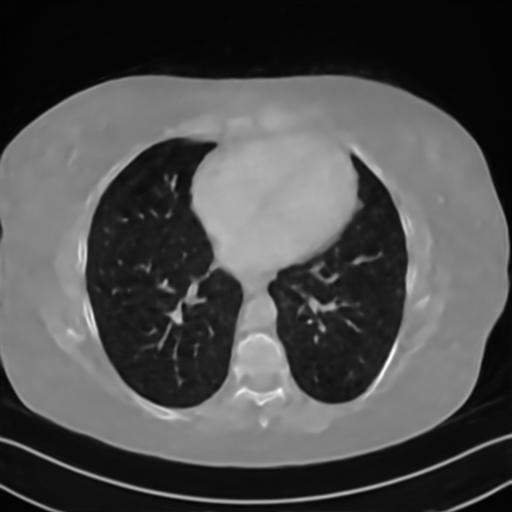}}
      \subfigure[GRAD]{
            \includegraphics[width=0.23\linewidth]{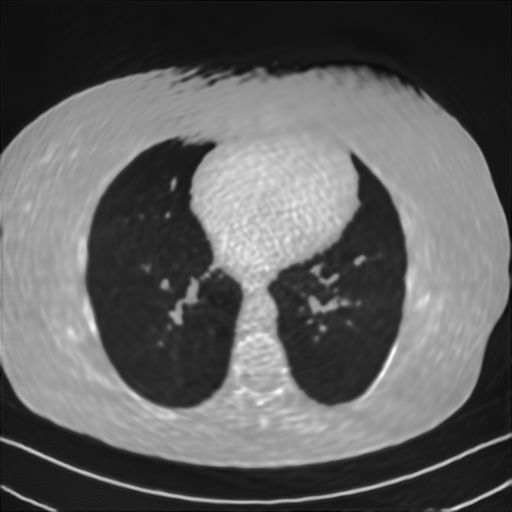}}
      \subfigure[LRPE]{
            \includegraphics[width=0.23\linewidth]{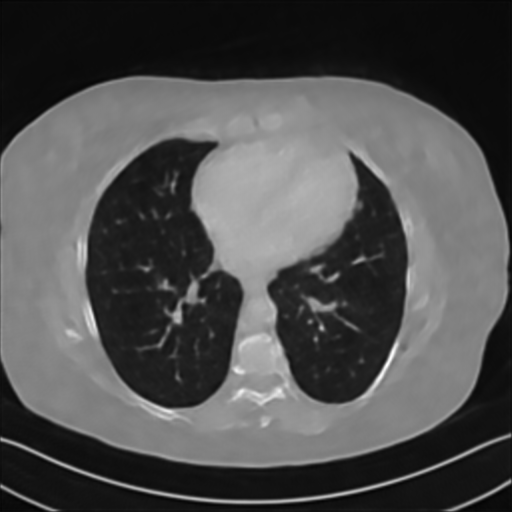}}
\caption{Limited-angle reconstruction experiment with $150^{\circ}$ scanning angular range and 5\% Gaussian noises.}
	\label{reconstruction150}
\end{figure*}

\subsection{Experiments on the limited-angle reconstruction}
In this subsection, we evaluate the performance of our LRPE model on the limited-angle reconstruction, where 5\% Gaussian noises are corrupted in the projection data. Both PSNR and SSIM of the comparison methods are provided in Table \ref{limite}. We can observe that the reconstruction qualities of all methods decrease as the scanning angle shrinks. Not surprisingly, our LRPE model has the numerical advantage compared to other comparison algorithms, which provides a $0.8 \mathrm{~dB}$ higher PSNR than the LRIP on $90^{\circ}$ limited-angle reconstruction task. Particularly, the projection data of the low-resolution prior used by LRIP model is different from ours, which was computed using the down-sampling matrix. Obviously, our setting is more reasonable and in accord with the CT scanner, which also gives better reconstruction results. Figure \ref{reconstruction150} and Figure \ref{90} present the reconstruction results with a scanning range of $150^{\circ}$ limited-angle and $90^{\circ}$ limited-angle, respectively. As can be observed, the learning-based methods outperform both FBP and TV, which contain noticeable artifacts in regions within missing angles. Moreover, our LRPE model surpasses the SIPID, PD, FSR, and LRIP by providing more image details and continuous contours. Thus, both quantitative and qualitative results confirm that the low-resolution image is a suitable prior for the ill-posed limited-angle reconstruction. Furthermore, by comparing the reconstruction results of 150$\circ$ limited-angle and $90^\circ$, it is not difficult to find that the resolution of the reconstruction results of all methods significantly decreases with the loss of angles. Although our method can provide better structural information by introducing the low-resolution image prior, it still suffers details missing due to angle deficient. It is mainly because our low-resolution image is the direct reconstruction result of 90$^\circ$ limited-angle data. Thus, our future work includes to explore effective reconstruction methods for low-resolution images.

In Figure \ref{iterate10}, we present and compare the output results of the LRIP, GRAD, and LRPE models during the iterative process. Although we have shown the numerical convergence of our LRPE model in the inference phase in Figure \ref{parameters} (a), the visual comparison further demonstrates that our model can provide more meaningful intermediate results in the iterative process. Compared to GRAD model, our method can effectively improve the quality of reconstruction results by introducing a low-resolution image prior. On the other hand, the results of the LRIP model are not stable in the first six stages, and only provide meaningful results in the last two stages. Thus, the deep equilibrium model with convergence guarantee is a better choice for establishing a stable end-to-end reconstruction network.

\begin{figure*}[t]
      \centering
      \subfigure[FBP]{
            \includegraphics[width=0.23\linewidth]{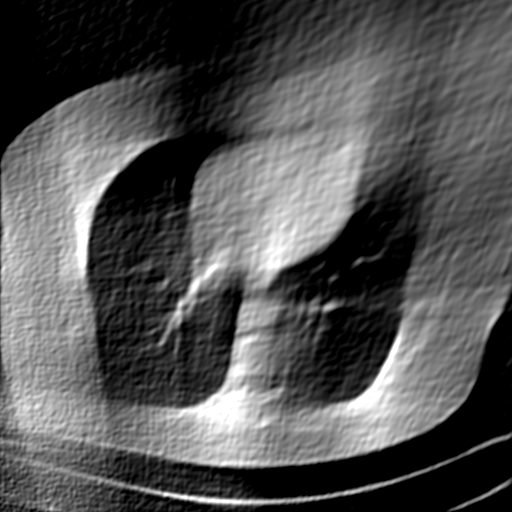}}
      \subfigure[TV]{
			\includegraphics[width=0.23\linewidth]{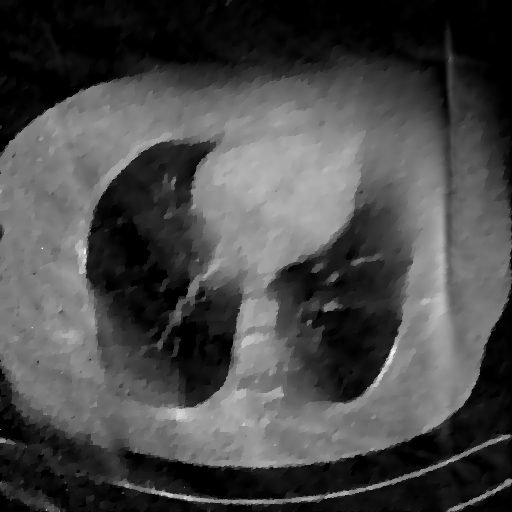}}
      \subfigure[SIPID]{
            \includegraphics[width=0.23\linewidth]{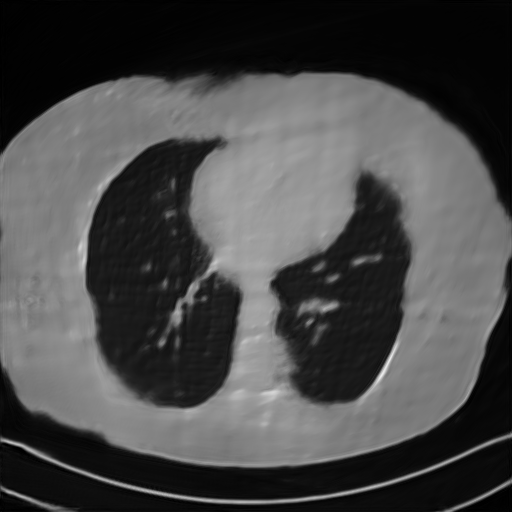}}
      \subfigure[PD]{
			\includegraphics[width=0.23\linewidth]{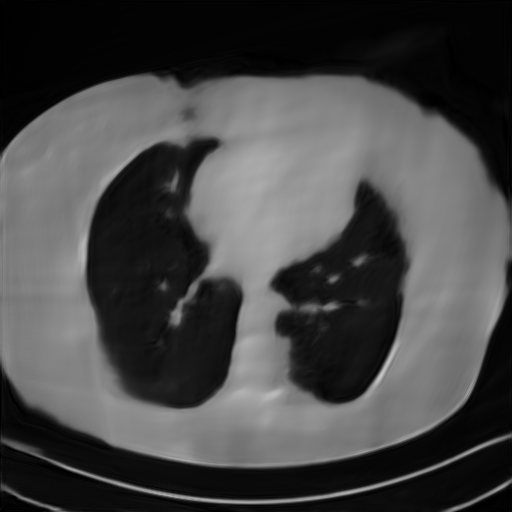}}\\
      \subfigure[SFSR]{
            \includegraphics[width=0.23\linewidth]{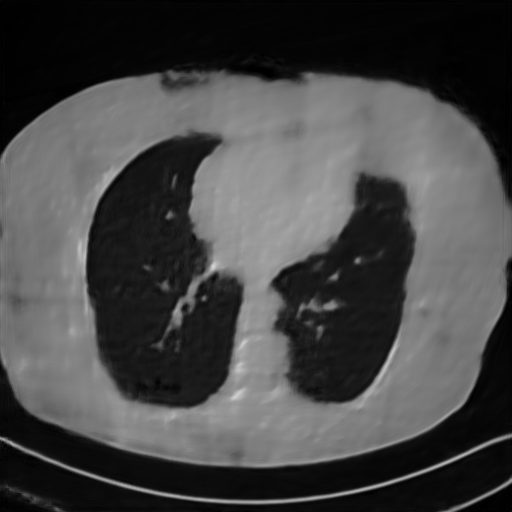}}
      \subfigure[LRIP]{
			\includegraphics[width=0.23\linewidth]{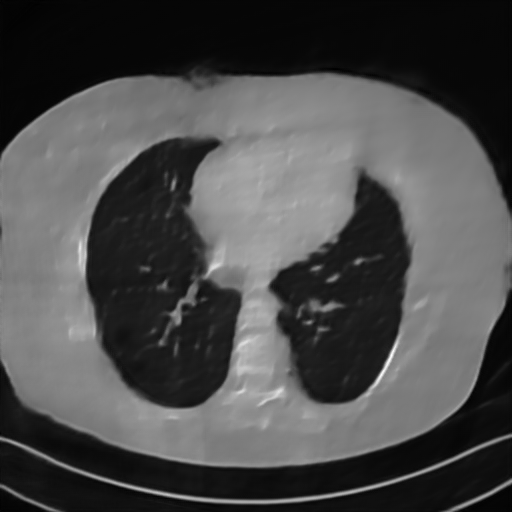}}
      \subfigure[GRAD]{
			\includegraphics[width=0.23\linewidth]{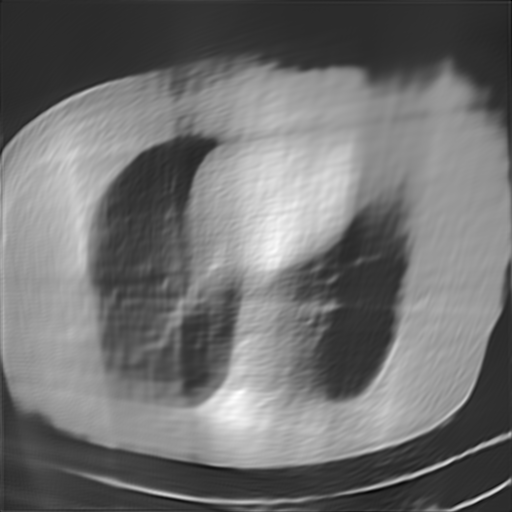}}
      \subfigure[LRPE]{
            \includegraphics[width=0.23\linewidth]{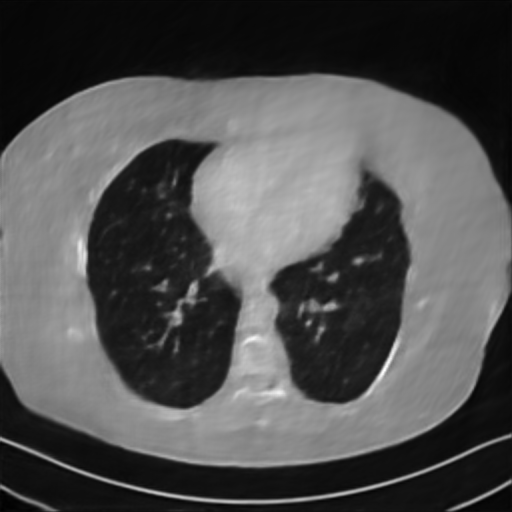}}
	\caption{Limited-angle reconstruction experiment of the AAPM phantom dataset with $90^{\circ}$ scanning angular range and 5\% Gaussian noises. }
	\label{90}
\end{figure*}

\begin{figure*}[htbp]
      \centering
      \subfigure[PSNR:9.1227]{
			\includegraphics[width=0.237\linewidth]{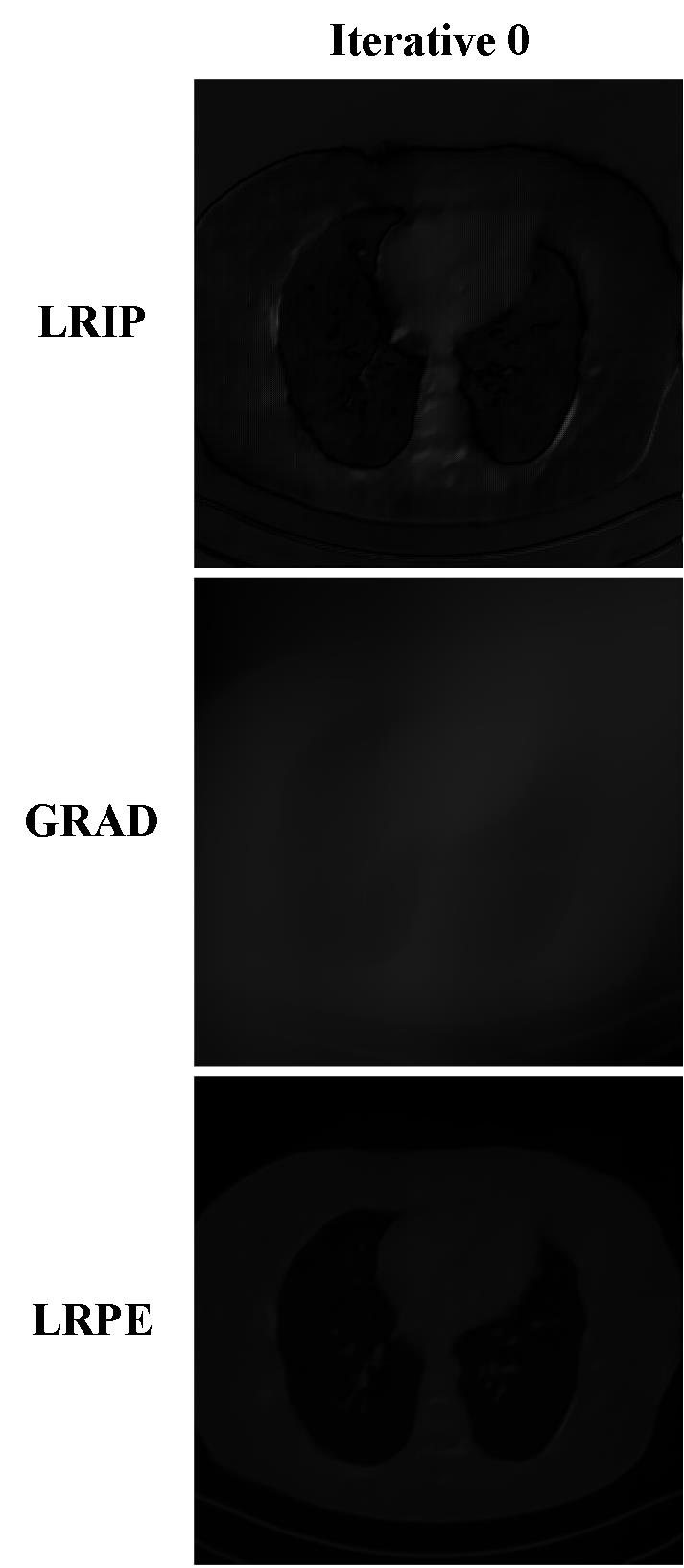}}
      \subfigure[PSNR:10.2367]{
			\includegraphics[width=0.17\linewidth]{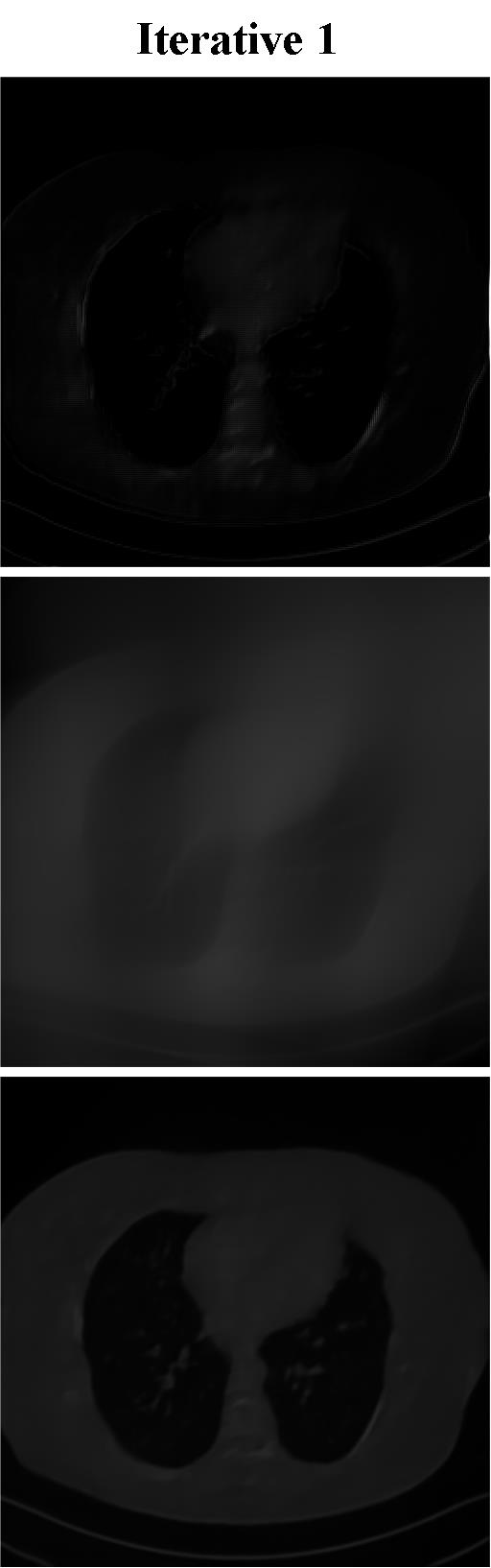}}
      \subfigure[PSNR:11.9091]{
			\includegraphics[width=0.17\linewidth]{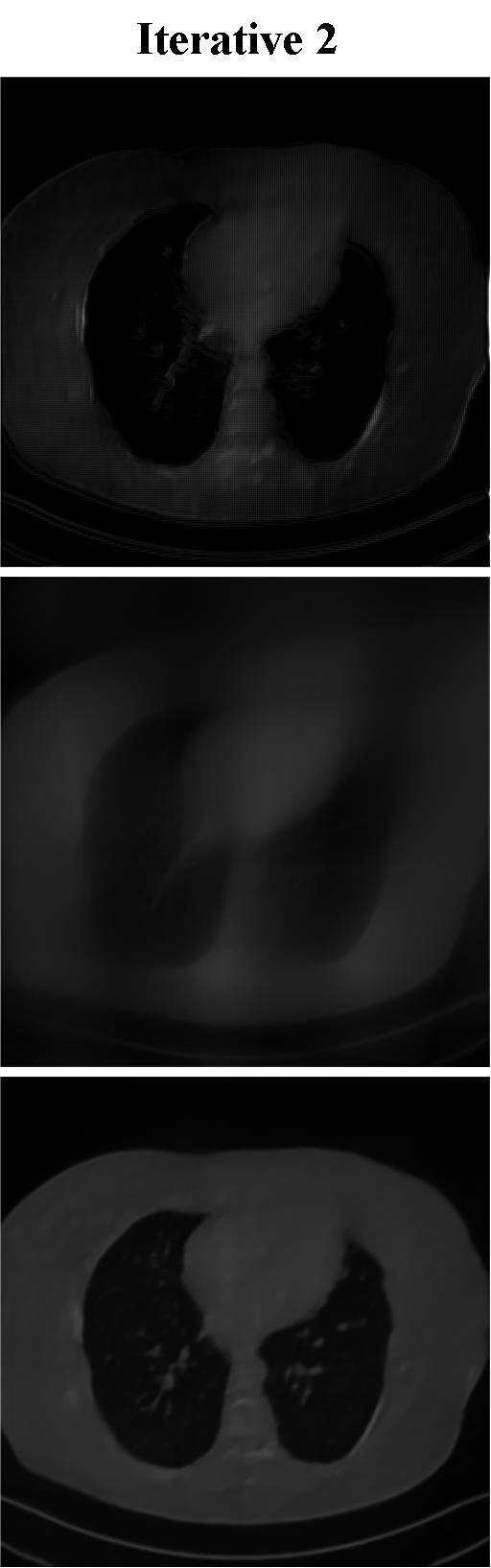}}
      \subfigure[PSNR:14.2079]{
			\includegraphics[width=0.17\linewidth]{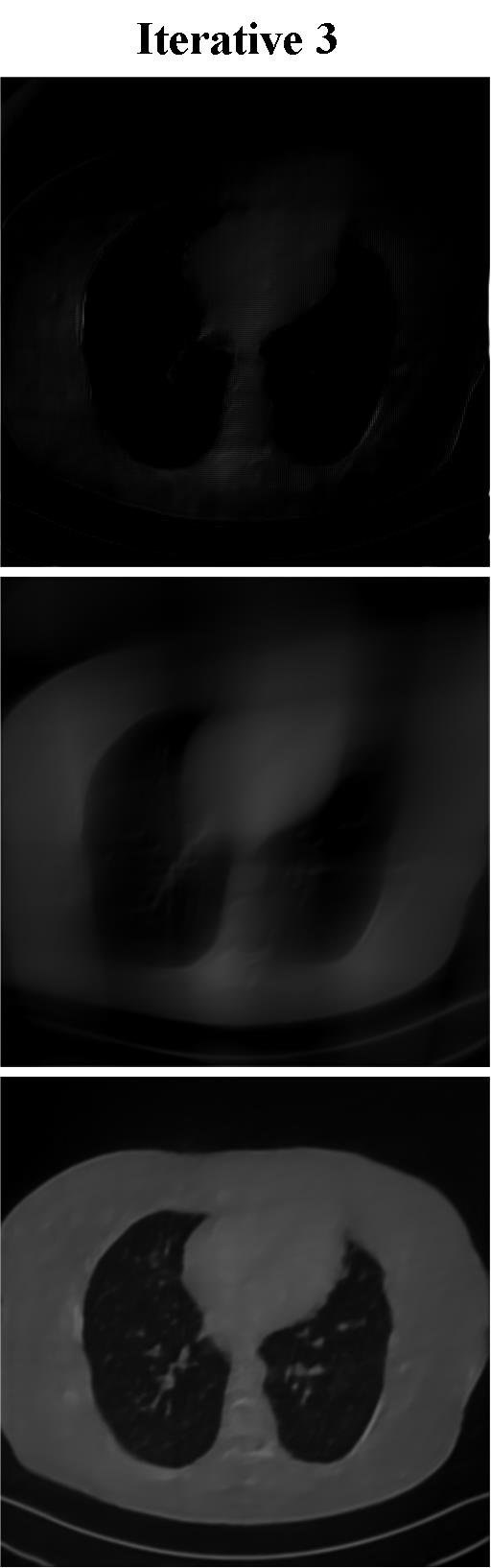}}
      \subfigure[PSNR:17.2330]{
			\includegraphics[width=0.17\linewidth]{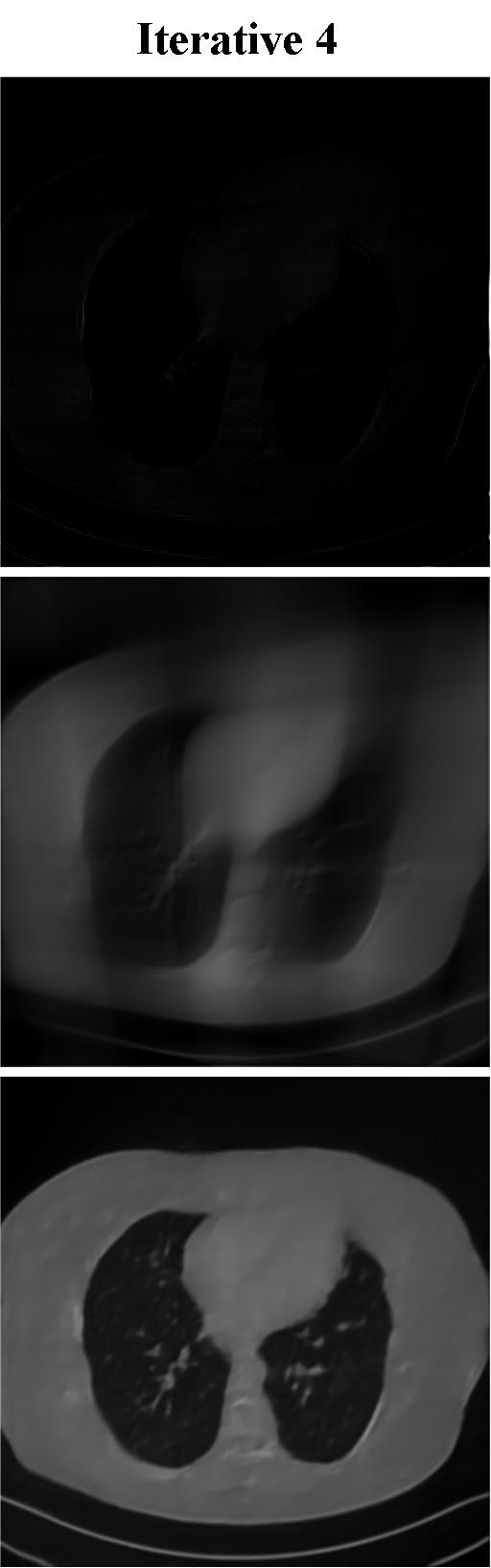}}\\
      \subfigure[PSNR:20.9317]{
			\includegraphics[width=0.237\linewidth]{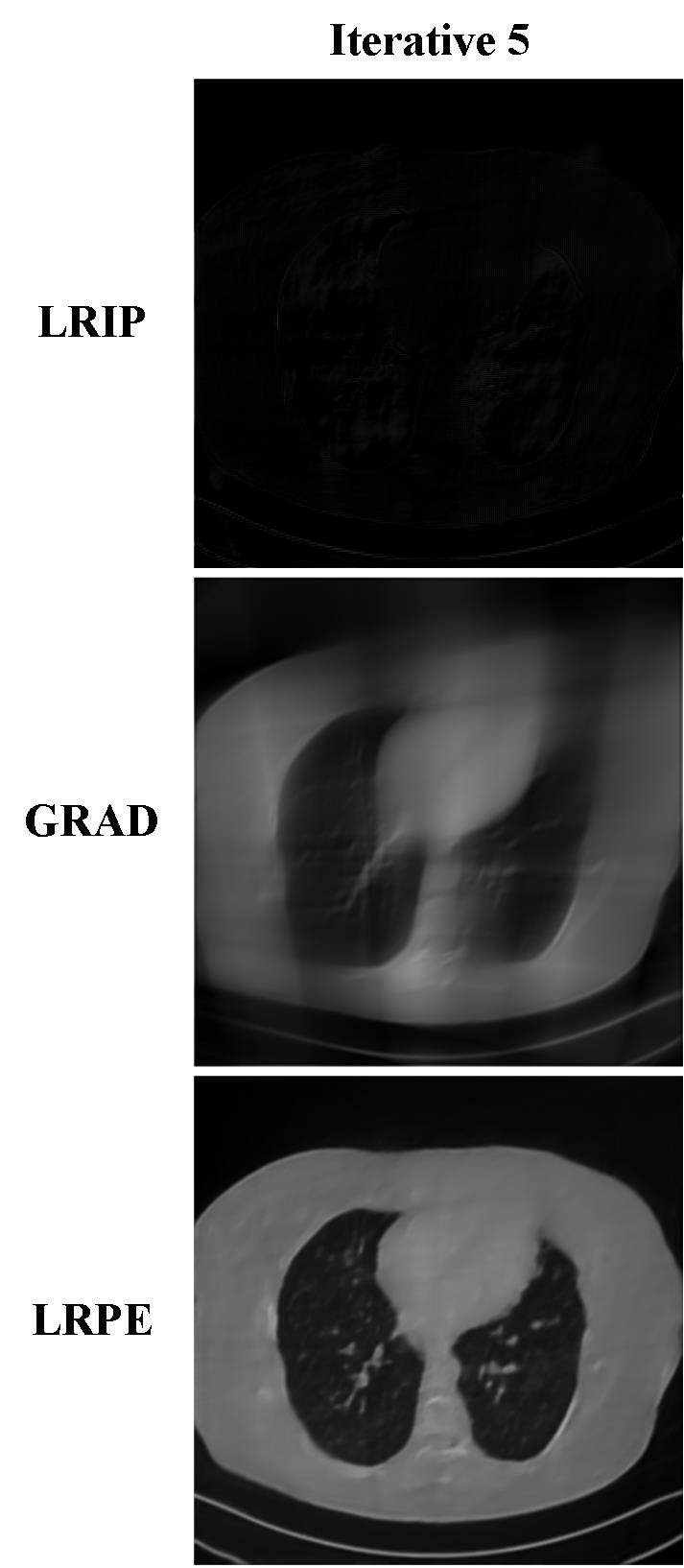}}
      \subfigure[PSNR:23.0121]{
			\includegraphics[width=0.17\linewidth]{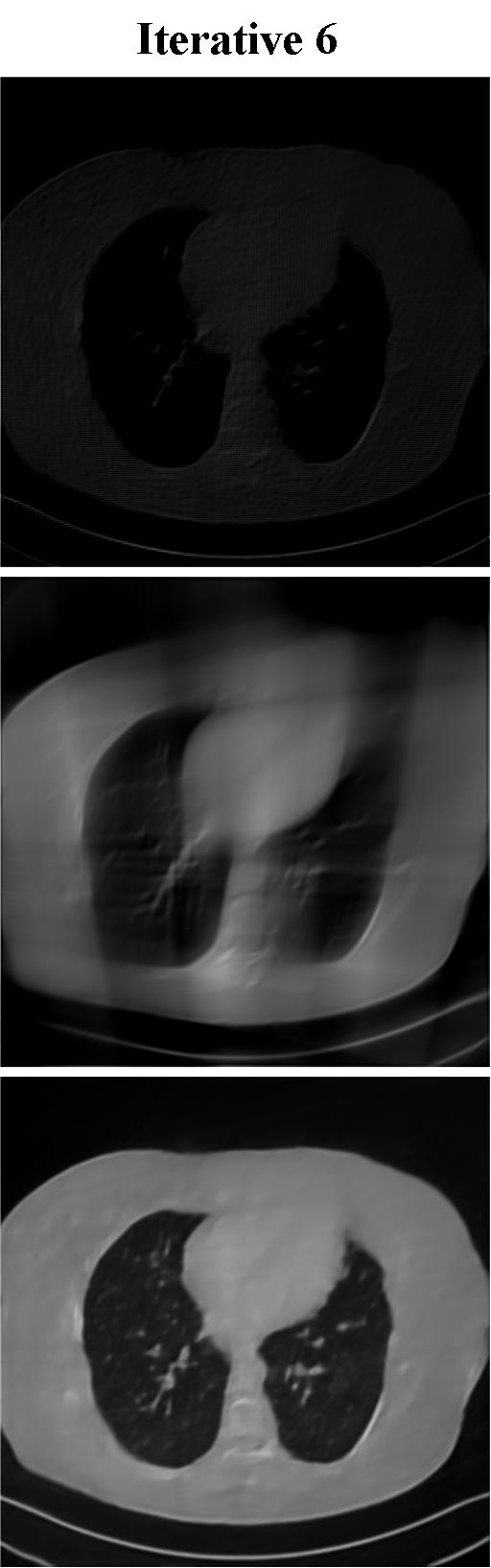}}
      \subfigure[PSNR:24.4391]{
			\includegraphics[width=0.17\linewidth]{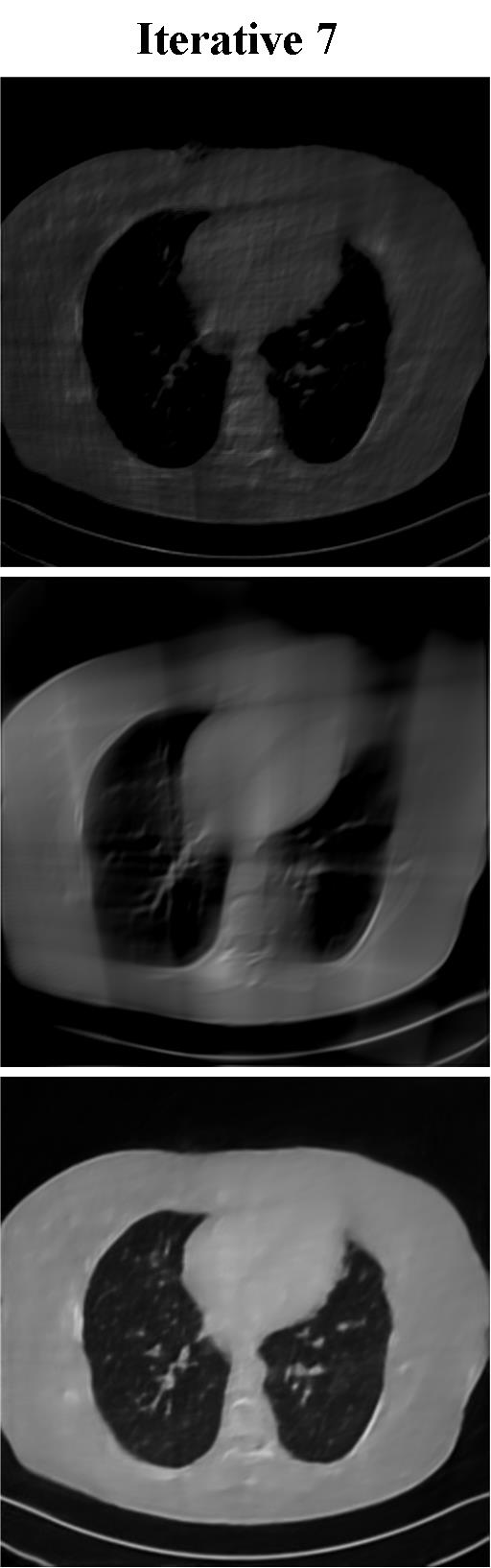}}
      \subfigure[PSNR:25.8657]{
			\includegraphics[width=0.17\linewidth]{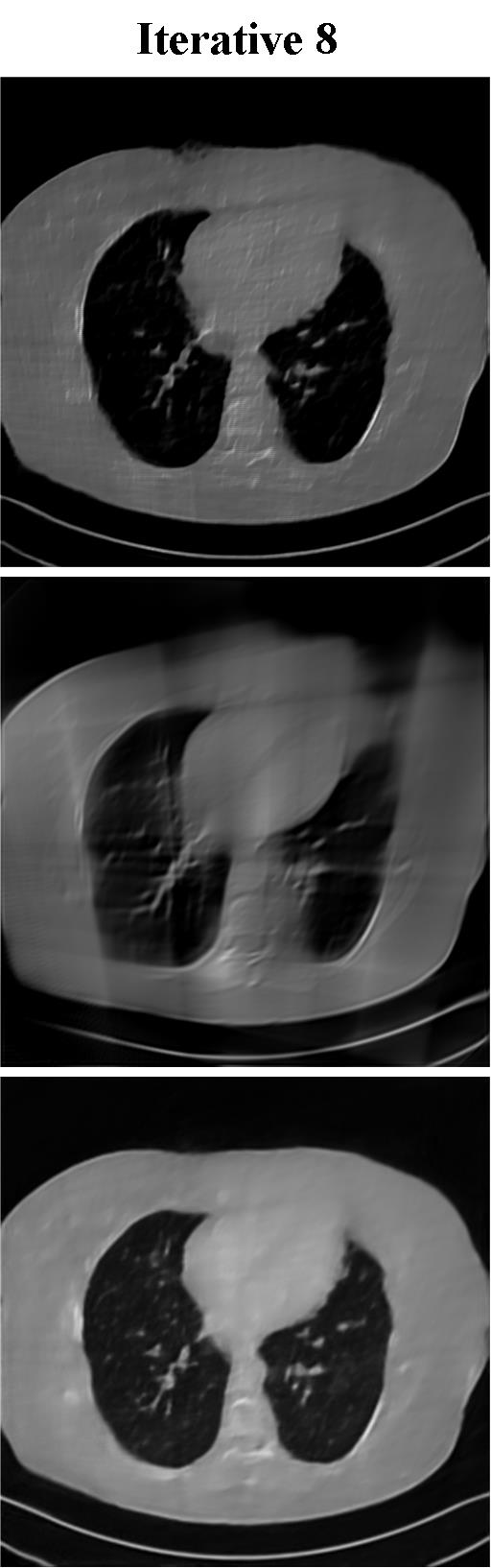}}
      \subfigure[PSNR:25.9517]{
			\includegraphics[width=0.17\linewidth]{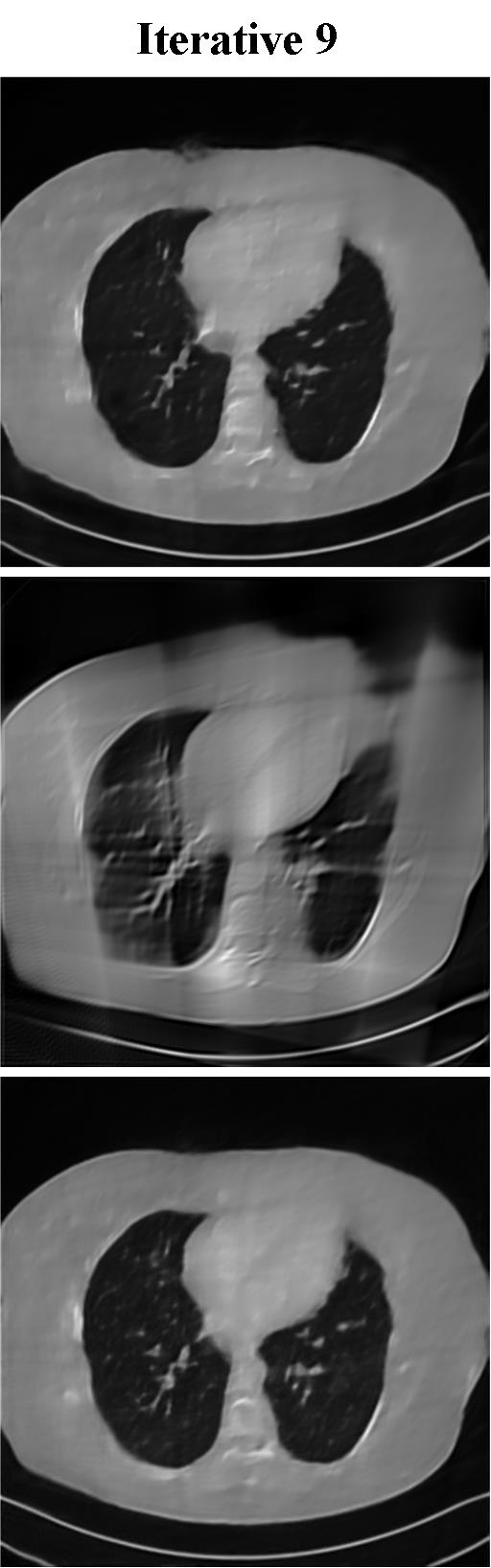}}
	\caption{The comparison results of the LRIP, GRAD, and LRPE model on the human phantom dataset with a scanning angular range of $90^{\circ}$ and 5\% Gaussian noise.}
	\label{iterate10}
\end{figure*}

\begin{figure*}[htbp]
      \centering
      \subfigure[Sparse-view reconstruction]{
			\includegraphics[width=0.48\linewidth]{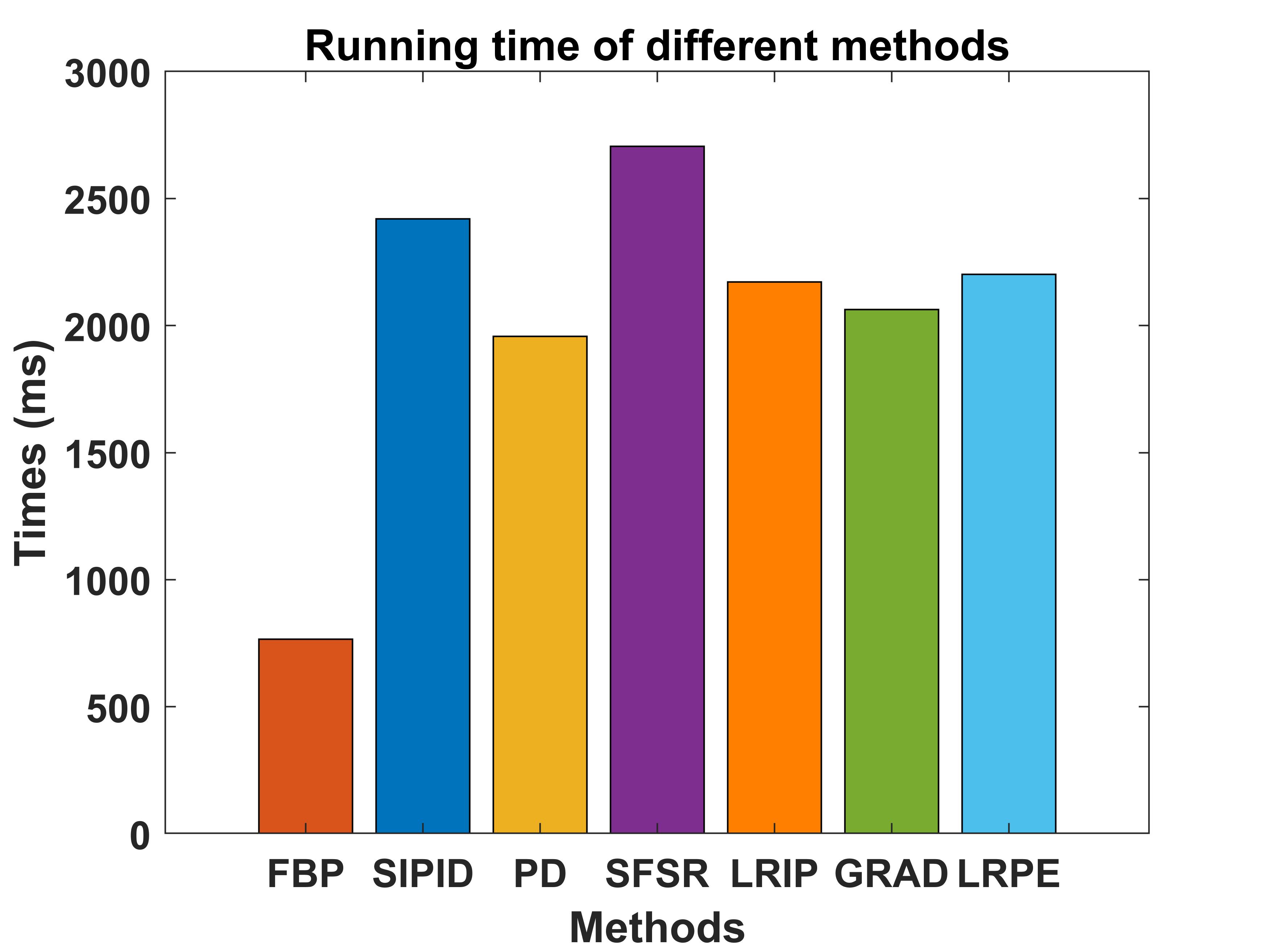}}
      \subfigure[Limited-angle reconstruction]{
			\includegraphics[width=0.48\linewidth]{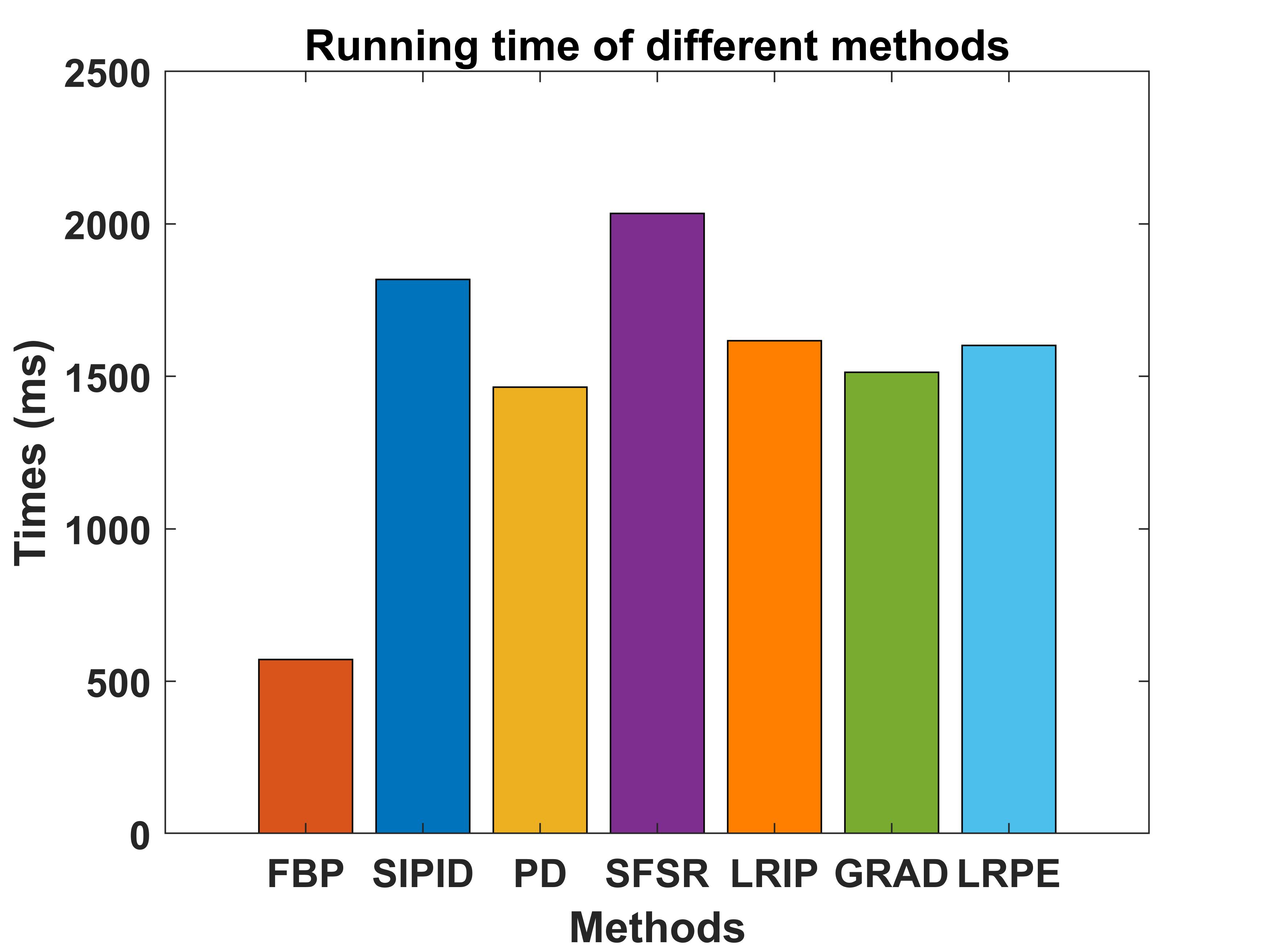}}
	\caption{The execution times of various methods were assessed using a sparse angle of 30 degrees and a limited scanning angular range of $90^{\circ}$.}
	\label{time}
\end{figure*}

Last but not least, we conduct the comparison study on the running time of the evaluated methods, which are illustrate in Figure \ref{time}. Notably, excluding the running time of the TV method, it exhibits an execution time approximately ten times longer than that of the FBP. By analyzing the bar chart results, we can observe that although our method is not as rapid as the PD method, we incur only a 100-millisecond increase in running time to achieve superior reconstruction results.

\section{Concluding Remarks}
In this paper, we proposed a novel end-to-end equilibrium imaging geometric model for the ill-posed image reconstruction problems. The low-resolution image was used prior to guarantee that the proposed model can achieve good reconstruction results based on incomplete projection data. The advantages of our proposal were validated using a substantial amount of numerical experiments on sparse-view and limited-angle settings. The experimental results demonstrated that our model can enable accurate and efficient CT image reconstruction. Although our downsampling imaging geometric modeling was studied for the fan-beam imaging systems, it can be also used for other modern CT settings. Similar works include \cite{he2021downsampled} and \cite{gao2022lrip}, where the low-resolution image prior was incorporated into the reconstruction network.
The prior information in \cite{he2021downsampled} was obtained by halving the receiver signal, requiring a second scan of the patient, which is clinically prohibitive. In contrast, our approach eliminates the need of re-scanning patients. On the other hand, Gao \emph{et al.} \cite{gao2022lrip} also requires two distinct sets of projection data. Obviously, the low-resolution image obtained from the same scan can better maintain consistency with the high-resolution image, thereby preserving more structural and fine information. According to the experimental results, although our LRPE model is superior to other algorithms in reconstruction results, there is still a problem of detail loss in the case of incomplete data sampling. Thus, enhancing the quality of low-resolution images through effective sinogram restoration methods is a possible way to improve the reconstruction results of our LRPE model on severe incomplete data problems. One of the future research directions is to enhance the resolution of reconstructed images under extreme limited-angle and low-dose scanning problems. In theoretical analysis, the noises in data present new challenges in estimating the gradients of network models. We plan to investigate the convergence of the deep equilibrium model based on stochastic gradient similar to \cite{bottou2018optimization}.

\section{Acknowledgments}

The work was partially supported by the National Natural Science Foundation of China (NSFC 12071345).

\section*{Appendix}
\bibliographystyle{plain}
\bibliography{RefsDouble-resolution}

\end{document}